\documentclass[11pt]{article}

\usepackage{times}
\usepackage{latexsym}
\usepackage[margin=1in]{geometry}
\usepackage{graphicx}
\usepackage{amsthm}
\usepackage{amsmath}
\usepackage{amssymb}
\usepackage{amsfonts}
\usepackage[linesnumbered]{algorithm2e}
\usepackage{algpseudocode}
\usepackage{float}
\usepackage{fullpage}
\usepackage{mathptmx}
\usepackage{varwidth}
\usepackage{tabu}
\usepackage{xspace}
\usepackage{xcolor}
\usepackage{enumerate}
\usepackage{enumitem}
\newcommand{\eat}[1]{}

\newcommand{\mr}{\text{\sc mr}\xspace}
\newcommand{\sol}{\text{\sc sol}\xspace}
\newcommand{\alg}{\text{\sc alg}\xspace}
\newcommand{\opt}{\text{\sc opt}\xspace}
\newcommand{\apx}{\text{\sc apx}\xspace}
\newcommand{\mrs}{\text{\sc mrs}\xspace}

\newtheorem{theorem}{Theorem}
\newtheorem{lemma}[theorem]{Lemma}

\newtheorem{observation}[theorem]{Observation}
\newtheorem{corollary}[theorem]{Corollary}

\DeclareMathOperator*{\argmin}{argmin}

\newcommand{\bd}{\textbf{d}}
\newcommand{\bx}{\textbf{x}}
\newcommand{\by}{\textbf{y}}

\newcommand{\mst}{{\sc mst}\xspace}
\newcommand{\tsp}{{\sc tsp}\xspace}
\newcommand{\stt}{{\sc stt}\xspace}

\newcommand{\I}{{\bf I}}
\newcommand{\poly}{{\bf P}\xspace}
\newcommand{\np}{{\bf NP}}

\title{Robust Algorithms for TSP and Steiner Tree\thanks{An extended abstract of this paper appeared in the Proceedings of the $47^{th}$ {\em International Colloquium on Automata, Languages, and Programming {\rm \bf (ICALP)}}, 2020.}}
\author{
  Arun Ganesh
  \footnote{Department of Electrical Engineering and Computer Sciences, UC Berkeley. Email: \texttt{arunganesh@berkeley.edu}. Supported in part by NSF Award CCF-1535989.}
  \and
  Bruce M. Maggs
  \footnote{Department of Computer Science, Duke University, and Emerald Innovations. Email: \texttt{bmm@cs.duke.edu}. Supported in part by NSF Award CCF-1535972.}
  \and 
  Debmalya Panigrahi
  \footnote{Department of Computer Science, Duke University. Email: \texttt{debmalya@cs.duke.edu}. Supported in part by NSF grants CCF-1535972, CCF-1955703,  an NSF CAREER Award CCF-1750140, and the Indo-US Virtual Networked Joint Center on Algorithms under Uncertainty.}\\ 
  }
\date{}

\begin{document}

\maketitle

\begin{abstract}
    Robust optimization is a widely studied area in operations research, where the algorithm takes as input a range
    of values and outputs a single solution that performs well for the entire range.
    Specifically, a robust algorithm aims to minimize {\em regret}, defined as the maximum difference between the 
    solution's cost and that of an optimal solution in hindsight once the input has been realized. For graph problems 
    in \poly, such as shortest path and minimum spanning tree, robust polynomial-time algorithms that obtain a constant approximation
    on regret are known. 
    In this paper, we study robust algorithms for minimizing regret in \np-hard graph optimization problems, and give 
    constant approximations on regret for the classical traveling salesman and Steiner tree problems.
\end{abstract}

\thispagestyle{empty}
\setcounter{page}{0}
\clearpage

\section{Introduction}
\label{sec:intro}
In many graph optimization problems, the inputs are not known precisely and 
the algorithm is desired to perform well over a range of inputs. For 
instance, consider
the following situations. Suppose we are planning the
delivery route of a vehicle that must deliver goods to $n$ locations. 
Due to varying traffic conditions, the exact travel times between locations
are not known precisely, but a range of possible travel times is available
from historical data. Can we design a tour that is nearly optimal for {\em all}
travel times in the given ranges? Consider another situation where 
we are designing a telecommunication 
network to connect a set of locations. We are given cost estimates on connecting
every two locations in the network but these estimates might be off 
due to unexpected construction problems. Can we design the network
in a way that is nearly optimal for {\em all} realized construction costs?

These questions have led to the field of {\em robust} graph algorithms.
Given a range of weights $[\ell_e, u_e]$ for every edge $e$, the goal is to find a solution that 
minimizes {\em regret}, defined as the maximum difference between the algorithm's cost and 
the optimal cost for any edge weights. In other words, the goal is to obtain:
$\min_{\sol} \max_{\I} (\sol(\I) - \opt(\I))$, where $\sol(\I)$ (resp. $\opt(\I)$) denotes the cost of $\sol$ (resp. the optimal solution) in instance $\I$, $\sol$ ranges over all feasible
solutions, and $\I$ ranges over all realizable inputs. We emphasize that $\sol$ is a fixed solution 
(independent of $\I$) whereas the solution determining $\opt(\I)$ is dependent on the input $\I$.
 The solution that achieves this minimum is called the {\em minimum regret solution} ($\mrs$), 
 and its regret is the {\em minimum regret} ($\mr$). In many cases, however, minimizing regret
 turns out to be \np-hard, in which case one seeks an approximation guarantee. Namely,
 a $\beta$-approximation algorithm satisfies, for all input realizations $\I$, 
 $\sol(\I) - \opt(\I) \leq \beta \cdot \mr$, i.e., $\sol(\I) \leq \opt(\I) + \beta \cdot \mr$.

It is known that minimizing regret is \np-hard for shortest path~\cite{Zielinski04} and 
minimum cut~\cite{AissiBV08} problems, and using a general theorem for 
converting exact algorithms to robust ones, $2$-approximations are
known for these problems~\cite{Conde12,KasperskiZ06}. In some cases, better results are known 
for special classes of graphs, e.g., \cite{KasperskiZ07}. Robust minimum spanning tree (\mst) has also been 
studied, although in the context of making exponential-time exact 
algorithms more practical~\cite{YamanKP01}. Moreover, robust optimization has been extensively researched for other (non-graph) problem domains in the operations 
research community, and has led to results in clustering \cite{Averbakh05, AverbakhB97,AverbakhB00,KouvelisY97}, 
linear programming \cite{InuiguchiM95,MausserL98}, and other areas \cite{Averbakh01,KasperskiZ06}.   More details can be found in the book by Kouvelis and Yu~\cite{KouvelisY96} and the survey by Aissi {\em et al.}~\cite{AISSI2009427}.


To the best of our knowledge, all previous work in polynomial-time algorithms for minimizing regret in robust graph optimization focused 
on problems in \poly. In this paper, we study robust graph algorithms for minimizing regret in \np-hard optimization problems. In particular, 
we study robust algorithms for the classical traveling salesman (\tsp) 
and Steiner tree (\stt) problems, that model e.g. the two scenarios described at the beginning
of the paper. 
As a consequence of the \np-hardness, we cannot hope to show 
guarantees of the form: $\sol(\I) \leq \opt(\I) + \beta \cdot \mr$, since for $\ell_e = u_e$
(i.e., $\mr = 0$), this would imply an exact algorithm for an \np-hard optimization problem.
Instead, we give guarantees: $\sol(\I) \leq \alpha \cdot \opt(\I) + \beta \cdot \mr$,
where $\alpha$ is (necessarily) at least as large as the best approximation guarantee for the optimization problem.
We call such an algorithm an $(\alpha, \beta)$-robust algorithm.
If both $\alpha$ and $\beta$ are constants, we call it a constant-approximation 
to the robust problem.
In this paper, our main results are constant approximation algorithms for the robust traveling
salesman and Steiner tree problems. 
We hope that our work will lead to further research in the field of robust approximation algorithms, 
particularly for minimizing regret in other \np-hard optimization problems in graph algorithms as well as in other domains.

\eat{
To the best of our knowledge, efficient robust algorithms for \np-hard optimization problems
have not been studied before our work. 
Naturally, the motivation for robust algorithms for \np-hard problems is exactly the same as 
that for problems in \poly. However, unlike
for problems in \poly, there is no black box reduction from approximation algorithms
for \np-hard problems to their robust counterparts. Indeed, replacing an exact algorithm with 
an approximate one in the reduction framework can yield arbitrarily bad solutions in general 
(we give an example in Appendix~\ref{section:deferred-tsp}). 
This makes the question of 
designing robust algorithms for \np-hard problems a particularly interesting one, and we 
initiate this line of research by studying two \np-hard graph optimization problems that 
arguably are the most natural candidates for the robust optimization framework. 
}

\subsection{Problem Definition and Results}
We first define the Steiner tree (\stt) and traveling salesman problems (\tsp).
In both problems, the input is an undirected graph $G = (V, E)$ with non-negative 
edge costs. In Steiner tree, we are also given a subset of vertices called 
{\em terminals} and the goal is to obtain a minimum cost connected subgraph of $G$ 
that spans all the terminals. In traveling salesman, the goal
is to obtain a minimum cost tour that visits every vertex in $V$\footnote{There are two common and equivalent assumptions made in the \tsp literature in order to achieve reasonable approximations. In the first assumption, the algorithms can visit vertices multiple times in the tour, while in the latter, the edges satisfy the metric property. We use the former in this paper.}. In the robust
versions of these problems, the edge costs are ranges $[\ell_e, u_e]$ 
from which any cost may realize.

Our main results are the following:
\begin{theorem}\label{thm:constantapprox}(Robust Approximations.)
There exist constant approximation algorithms for the robust traveling salesman and Steiner tree problems. 
\end{theorem}

\noindent{\bf Remark:}
The constants we are able to obtain for the two problems 
are very different: $(4.5, 3.75)$ for \tsp (in Section~\ref{section:tsp}) 
and $(2755, 64)$ for \stt (in Section~\ref{section:steiner}).  While we did not attempt
to optimize the precise constants, obtaining small constants for \stt comparable to
the \tsp result requires new ideas beyond our work and is an interesting open problem.



We complement our algorithmic results with lower bounds. Note that 
if $\ell_e = u_e$, we have $\mr = 0$ and thus an $(\alpha, \beta)$-robust 
algorithm gives an $\alpha$-approximation for precise inputs. So, hardness of 
approximation results yield corresponding lower bounds on $\alpha$.
More interestingly, we show that hardness of approximation results also yield lower bounds 
on the value of $\beta$ (see Section~\ref{sec:hardness} for details):
%
%
\begin{theorem}(APX-hardness.)
A hardness of approximation of $\rho$ for \tsp (resp., \stt) under $\poly\not=\np$ implies
that it is \np-hard to obtain
$\alpha \leq \rho$ (irrespective of $\beta$) and $\beta \leq \rho$ (irrespective of $\alpha$)
for robust \tsp (resp., robust \stt).
\end{theorem}



\eat{
There are two steps involved in any approximation algorithm that 
uses an LP relaxation: computing an optimal fractional solution and a rounding 
algorithm that preserves the optimality of the solution up to an approximation
factor loss. Typically, the first step is routine, either using a black box 
LP solver for a polynomial-size relaxation or using a simple separation oracle
for an exponential-size relaxation, and the ingenuity of the algorithm mostly 
resides in rounding the fractional solution. Interestingly, the situation is 
different for robust algorithms: we give a generic rounding algorithm that 
applies to a broad range of problems including \stt and \tsp, but most of the 
technical work lies in being able to obtain a separation oracle for the LP
relaxation.
}

\subsection{Our Techniques}
We now give a sketch of our techniques. Before doing so, we note that for problems in \poly with linear objectives, it is known that running an exact algorithm using weights $\frac{\ell_e + u_e}{2}$ gives a $(1, 2)$-robust solution~\cite{Conde12,KasperskiZ06}. One might hope that a similar result can be obtained for \np-hard problems by replacing the exact algorithm with an approximation algorithm in the above framework. Unfortunately, we show in Section~\ref{section:tsp} that this is not true in general. In particular, we give a robust \tsp instance where using a 2-approximation for \tsp with weights $\frac{\ell_e + u_e}{2}$ gives a solution that is {\bf not} $(\alpha, \beta)$-robust for {\em any} $\alpha = o(n), \beta = o(n)$. 
More generally, a black-box approximation run on a fixed realization could output a solution including edges that have small weight relative to $\opt$ for that realization (so including these edges does not violate the approximation guarantee), but these edges could have large weight relative to $\mr$ and $\opt$ in other realizations, ruining the robustness guarantee. This establishes a qualitative difference between robust approximations for problems in \poly considered earlier and \np-hard problems being considered in this paper, and demonstrates the need to develop new techniques for the latter class of problems.

\smallskip
\noindent
\textbf{LP relaxation.} We denote the input graph $G = (V, E)$.
For each edge $e\in E$, the input is a range $[\ell_e, u_e]$ where the actual edge weight
$d_e$ can realize to any value in this range. The robust version of a graph optimization problem is 
is then described by the LP
$$\min\{r : {\bf x} \in P; \sum_{e \in E} d_e x_e \leq \opt(\bd) + r,\ \forall \bd\},$$
where $P$ is the standard polytope for the optimization problem, 
and $\opt(\bd)$ denotes the cost of an optimal solution
when the edge weights are $\bd = \{d_e: e\in E\}$.
%
%
That is, this is the standard LP for the problem, but with the additional constraint that the fractional solution $\bx$ must have regret at most $r$ for any realization of edge weights. We call the additional constraints the {\em regret constraint set}. 
 Note that setting $\bx$ to be the indicator vector of \mrs and $r$ to $\mr$ gives a feasible solution to the LP; thus, the LP optimum is at most $\mr$. 

\smallskip
\noindent
\textbf{Solving the LP.}
We assume that the constraints in $P$ are separable in polynomial time 
(e.g., this is true for most standard optimization problems including \stt and \tsp).
So, designing the separation oracle comes down to separating the regret constraint set,
which requires checking that:
$$\max_{\bd} \left[\sum_{e \in E} d_ex_e - \opt(\bd)\right] =$$
$$\max_{\bd} \max_{\sol} \left[\sum_{e \in E} d_ex_e - \sol(\bd)\right] =  \max_{\sol} \max_{\bd} \left[\sum_{e \in E} d_ex_e - \sol(\bd)\right] \leq r.$$
Thus, given a fractional solution $\bx$, we need to find an integer solution $\sol$ and a weight vector $\bd$ that maximizes the regret of $\bx$ given by $\sum_{e \in E} d_ex_e - \sol(\bd)$. 
Once $\sol$ is fixed, finding $\bd$ that maximizes the regret is simple:  If $\sol$ does not include an edge $e$, then to maximize $\sum_{e \in E} d_e x_e - \sol(\bd)$, we set $d_e = u_e$; else if $\sol$ includes $e$, we set $d_e = \ell_e$. Note that in these two cases, edge $e$ contributes $u_e x_e$ and $\ell_e x_e - \ell_e$ respectively to the regret. The above maximization thus becomes: 
\begin{equation}
    \label{eq:omitted}
    \max_{\sol} \left[\sum_{e\notin \sol} u_e x_e + \sum_{e\in \sol} (\ell_e x_e - \ell_e)\right]
    = \sum_{e\in E} u_e x_e - \min_{\sol} \sum_{e\in \sol} (u_e x_e - \ell_e x_e + \ell_e).
\end{equation}
Thus, $\sol$ is exactly the optimal solution with edge weights $a_e := u_ex_e - \ell_ex_e + \ell_e$.
(For reference, we define the {\em derived} instance of the \stt problem as one with edge weights $a_e$.)

Now, if we were solving a problem in \poly, we would simply need to solve the problem on the derived instance. Indeed, we will show later that this yields an alternative technique for obtaining robust algorithms for problems in \poly, and recover existing results in~\cite{KasperskiZ06}.  However, we cannot hope to find an 
optimal solution to an \np-hard problem. Our first compromise is that we settle for 
an {\em approximate} separation oracle. More precisely, our goal is to show that there exists some 
fixed constants $\alpha', \beta' \geq 1$ such that if 
$\sum_e d_e x_e > \alpha' \cdot \opt(\bd) + \beta' \cdot r$ for some $\bd$, 
then we can find $\sol, \bd'$ such that $\sum_e d_e' x_e > \sol(\bd') + r$. Since the
LP optimum is at most $\mr$, we can then obtain an $(\alpha', \beta')$-robust {\em fractional} solution
using the standard ellipsoid algorithm. 

For \tsp, we show that the above guarantee can be achieved by the classic \mst-based $2$-approximation
on the derived instance.
The details appear in Section~\ref{section:tsp}. 
Although \stt also admits a $2$-approximation based on the \mst solution, this turns out to
be insufficient for the above guarantee. Instead, we use a different approach here. 
We note that the regret of the fractional solution against any fixed solution $\sol$
(i.e., the argument over which Eq.~\eqref{eq:omitted} maximizes)
can be expressed as the following difference:
$$\sum_{e \notin \sol} (u_ex_e - \ell_e x_e + \ell_e) - \sum_{e \in E} (\ell_e - \ell_e x_e)
= \sum_{e \notin \sol} a_e - \sum_{e \in E} b_e, \text{ where } b_e := \ell_e - \ell_e x_e.$$ 
The first term is the weight of edges in the derived instance that are {\em not} in \sol.
The second term corresponds to a new \stt instance with different edge weights $b_e$. 
It turns out that the overall problem now reduces to showing the following 
approximation guarantees on these two \stt instances ($c_1$ and $c_2$ are constants):\\
%
%
$$
\text{(i)} \sum_{e \in \alg \setminus \sol} a_e \leq c_1 \cdot \sum_{e \in \sol \setminus \alg} a_e
\qquad \text{and} \qquad
\text{(ii)} \sum_{e \in \alg} b_e \leq c_2 \cdot \sum_{e \in \sol} b_e.
$$
%
%
%
Note that guarantee (i) on the derived instance is an unusual ``difference approximation'' that is stronger than usual approximation guarantees. Moreover, we need these approximation bounds to {\em simultaneously} hold, i.e., hold for the same \alg. Obtaining these dual approximation bounds simultaneously forms the most technically challenging part of our work, and is given in Section~\ref{section:steiner}.

\smallskip
\noindent
\textbf{Rounding the fractional solution.}
After applying our approximate separation oracles, we have a fractional solution $\bx$ such that for all edge weights $\bd$, we have $\sum_e d_e x_e \leq \alpha' \cdot \opt(\bd) + \beta' \cdot \mr$. Suppose that, ignoring the regret constraint set, the LP we are using has integrality gap at most $\delta$ for precise inputs. Then a natural rounding approach is to search for an integer solution \alg that has minimum regret with respect to the specific solution $\delta \bx$, i.e., \alg satisfies:
\begin{equation}
    \label{eq:rounding}
    \alg = \argmin_{\sol} \max_{\bd} \left[\sol(\bd) - \delta \sum_{e \in E} d_e x_e\right].
\end{equation}
 Since the integrality gap is at most $\delta$, we have $\delta \cdot \sum_{e \in E} d_e x_e \geq \opt(\bd)$ for any $\bd$. This implies that:
 $$\mrs(\bd) - \delta \cdot \sum_{e \in E} d_e x_e  \leq \mrs(\bd) - \opt(\bd) \leq \mr.$$  
 Hence, the regret of \mrs with respect to $\delta x$ is at most $\mr$.  Since \alg has minimum regret with respect to $\delta \bx$, \alg's regret is also at most $\mr$.  Note that $\delta \bx$ is a $(\delta\alpha', \delta\beta')$-robust solution.  Hence, $\alg $ is a $(\delta\alpha', \delta\beta' + 1)$-robust solution. 

If we are solving a problem in \poly, finding $\alg$ that satisfies Eq.~\eqref{eq:rounding} is easy. So, using an integral LP formulation (i.e., integrality gap of 1), we get a $(1, 2)$-robust algorithm overall for these problems. This exactly matches the results in \cite{KasperskiZ06}, although we are using a different set of techniques. Of course, for \np-hard problems, finding a solution  $\alg$ that satisfies Eq.~\eqref{eq:rounding} is \np-hard as well. It turns out, however, that we can design a generic rounding algorithm that gives the following guarantee:

\begin{theorem}\label{thm:generalrounding}
There exists a rounding algorithm that takes as input an $(\alpha, \beta)$-robust fractional solution to \stt (resp. \tsp) and outputs a $(\gamma \delta \alpha, \gamma \delta \beta + \gamma)$-robust integral solution, where $\gamma$ and $\delta$ are respectively the best approximation factor and integrality gap for (classical) \stt (resp., \tsp).
\end{theorem}

We remark that while we stated this rounding theorem for \stt and \tsp here, we actually give a more general version (Theorem~\ref{thm:rounding}) in Section~\ref{sec:general} that applies to a broader class of covering problems including set cover, survivable network design, etc. and might be useful in future research in this domain.

\subsection{Related Work}

We have already discussed the existing literature in robust optimization for minimizing regret. 
Other robust variants of graph optimization have also been studied in the literature. In the \textit{robust combinatorial optimization} model proposed by Bertsimas and Sim~\cite{BertsimasS2003}, edge costs are given as ranges as in this paper, but instead of optimizing for all realizations of costs within the ranges, the authors consider a model where at most $k$ edge costs can be set to their maximum value and the remaining are set to their minimum value. The objective is to minimize the maximum cost over all realizations. 
In this setting, there is no notion of regret and an approximation algorithm for the standard problem translates to an approximation algorithm for the robust problem with the same approximation factor. 
 
In the \textit{data-robust model}~\cite{DhamdhereGRS05}, the input includes a polynomial number of explicitly defined ``scenarios'' for edge costs, with the goal of finding a solution that is approximately optimal for all given scenarios. 
That is, in the input one receives a graph and a polynomial number of scenarios $\bd^{(1)}, \bd^{(2)} \ldots \bd^{(k)}$ and the goal is to find $\alg$ whose maximum cost across all scenarios is at most some approximation factor times $\min_{\sol} \max_{i \in [k]} \sum_{e \in \sol} d_{e}^{(i)}$. In contrast, in this paper, we have exponentially many scenarios and look at the maximum of $\alg(\bd) - \opt(\bd)$ rather than $\alg(\bd)$. 
A variation of this is the \textit{recoverable robust model}~\cite{ChasseinG15}, where after seeing the chosen scenario, the algorithm is allowed to ``recover'' by making a small set of changes to its original solution.

Dhamdhere {\em et al.}~\cite{DhamdhereGRS05} also studies the \textit{demand-robust model}, where edge costs are fixed but the different scenarios specify different connectivity requirements of the problem. 
The algorithm now operates in two phases: In the first phase, the algorithm builds a partial solution $T'$ and then one of the scenarios (sets of terminals) $T_i$ is revealed to the algorithm. In the second phase, the algorithm then adds edges to $T'$ to build a solution $T$, but must pay a multiplicative cost of $\sigma_k$ on edges added in the second phase. The demand-robust model was inspired by a two-stage stochastic optimization model studied in, e.g., \cite{SwamyS06,ShmoysS06,SwamyS12,
DhamdhereGRS05,FeigeJMM07,KhandekarKMS13,GuptaNR14,GuptaNR16, GuptaPRS04,CharikarCP05} where the scenario is chosen according to a distribution rather than an adversary. 

Another related setting to the data-robust model is that of {\em robust network design}, introduced to model uncertainty in the demand matrix of network design problems (see the survey by Chekuri~\cite{Chekuri07}). This included the well-known VPN conjecture (see, e.g., \cite{GuptaKKRY01}), which was eventually settled in \cite{GoyalOS13}.
In all these settings, however, the objective is to minimize the maximum cost over all realizations, whereas in this paper, our goal is to minimize the maximum {\em regret} against the optimal solution. 
\subsection{Roadmap}

We present the general rounding algorithm for robust problems in Section~\ref{sec:general}. In Section~\ref{section:tsp}, we use this 
rounding algorithm to give a robust algorithm for the 
Traveling Salesman problem. 
Section~\ref{section:localsearch} gives a local search algorithm
for the Steiner Tree problem. Both the local search algorithm and the rounding 
algorithm from Section~\ref{sec:general} are then used to give a
robust algorithm for the Steiner Tree problem in Section~\ref{section:steiner}.
The hardness results for robust problems appear in Section~\ref{sec:hardness}.
Finally, we conclude with some interesting directions of future work in 
Section~\ref{sec:conclusion}.
\section{A General Rounding Algorithm for Robust Problems}\label{sec:general}

In this section we give the rounding algorithm of Theorem~\ref{thm:generalrounding}, which is a corollary of the following, more general theorem:

\begin{theorem}\label{thm:rounding}
Let $\mathcal{P}$ be an optimization problem defined on a set system $\mathcal{S} \subseteq 2^{E}$ that seeks to find the set $S\in \mathcal{S}$ that minimizes $\sum_{e \in S} d_e$, i.e., the sum of the weights of elements in $S$. In the robust version of this optimization problem, we have $d_e\in [\ell_e, u_e]$ for all $e\in E$.

Consider an LP formulation of $\mathcal{P}$ (called $\mathcal{P}$-LP) given by: $\{\min \sum_{e \in E} d_e x_e: \bx \in X, \bx \in  [0, 1]^E\}$, where $X$ is a polytope containing the indicator vector $\chi_{S}$ of all $S \in \mathcal{S}$ and not containing $\chi_{S}$ for any $S \notin \mathcal{S}$.
 The corresponding LP formulation for the robust version (called $\mathcal{P}_{\rm robust}$-LP) is given by: $\{\min r: \bx \in X, \bx \in  [0, 1]^E, \sum_{e \in E} d_e x_e \leq \opt(\bd) + r ~\forall \bd\}$.

Now, suppose we have the following properties:
\begin{itemize}
\item There is a $\gamma$-approximation algorithm for $\mathcal{P}$.
\item The integrality gap of $\mathcal{P}$-LP is at most $\delta$.
\item There is a feasible solution $\bx^*$ to $\mathcal{P}$-LP that satisfies:
$\forall \bd: \sum_{e \in E}d_e x_e^* \leq \alpha \cdot \opt(\bd) + \beta \cdot \mr$.
\end{itemize}
Then, there exists an algorithm that outputs a $(\gamma \delta \alpha, \gamma \delta \beta + \gamma)$-robust $\sol$ for $\mathcal{P}$.
\end{theorem}
\begin{proof}
The algorithm is as follows: Construct an instance of $\mathcal{P}$ which uses the same set system $\mathcal{S}$ and where element $e$ has weight $\max\{u_e(1 - \delta x_e^*), \ell_e (1 - \delta x_e^*)\} + \delta \ell_e x_e^*$. Then, use the $\gamma$-approximation algorithm for $\mathcal{P}$ on this instance to find an integral solution $S$, and output it.

Given a feasible solution $S$ to $\mathcal{P}$, note that:

$$\max_\bd[\sum_{e \in S} d_e - \delta \sum_{e \in E} d_e x_e^*] = \sum_{e \in S} \max\{u_e(1 - \delta x_e^*), \ell_e(1 - \delta x_e^*)\} - \sum_{e \notin S}\delta \ell_e x_e^*$$
$$= \sum_{e \in S} [\max\{u_e(1 - \delta x_e^*), \ell_e(1 - \delta x_e^*)\} + \delta \ell_e x_e^*] - \sum_{e \in E} \delta \ell_e x_e^*.$$

Now, note that since $S$ was output by a $\gamma$-approximation algorithm, for any feasible solution $S'$:
$$\sum_{e \in S} [\max\{u_e(1 - \delta x_e^*), \ell_e(1 - \delta x_e^*)\} + \delta \ell_e x_e^*] \leq \gamma \sum_{e \in S'} [\max\{u_e(1 - \delta x_e^*), \ell_e(1 - \delta x_e^*)\} + \delta \ell_e x_e^*] \implies$$
\begin{align*}
& \sum_{e \in S} [\max\{u_e(1 - \delta x_e^*), \ell_e(1 - \delta x_e^*)\} + \delta \ell_e x_e^*] - \gamma \sum_{e \in E} \delta \ell_ex_e^*\\
\leq & \gamma [\sum_{e \in S'} [\max\{u_e(1 - \delta x_e^*), \ell_e(1 - \delta x_e^*)\} + \delta \ell_e x_e^*] - \sum_{e \in E} \delta \ell_e x_e^*]\\
=&\gamma \max_\bd[\sum_{e \in S'} d_e - \delta \sum_{e \in E} d_e x_e^*].
\end{align*}

Since $\mathcal{P}$-LP has integrality gap $\delta$, for any fractional solution $\bx$, $\forall \bd: \opt(\bd) \leq \delta \sum_{e \in E} d_e x_e$. Fixing $S'$ to be the set of elements used in the minimum regret solution then gives:

$$\max_\bd[\sum_{e \in S'} d_e - \delta \sum_{e \in E} d_e x_e^*] \leq \max_\bd[\mrs(\bd) - \opt(\bd)] = \mr.$$

Combined with the previous inequality, this gives:
\begin{align*}
&\sum_{e \in S} [\max\{u_e(1 - \delta x_e^*), \ell_e(1 - \delta x_e^*)\} + \delta \ell_e x_e^*] - \gamma \sum_{e \in E} \delta \ell_ex_e^* \leq \gamma \mr \\
\implies&\sum_{e \in S} [\max\{u_e(1 - \delta x_e^*), \ell_e(1 - \delta x_e^*)\} + \delta \ell_e x_e^*] - \sum_{e \in E} \delta \ell_ex_e^* \leq \gamma \mr + (\gamma - 1) \sum_{e \in E} \delta \ell_ex_e^* \\
\implies&\max_\bd[\sum_{e \in S} d_e - \delta \sum_{e \in E} d_e x_e^*] \leq \gamma \mr + (\gamma - 1) \sum_{e \in E} \delta \ell_ex_e^*.
\end{align*}

This implies:
\begin{align*}
\forall \bd: \sol(\bd) = \sum_{e \in S} d_e &\leq \delta \sum_{e \in E} d_e x_e^* + \gamma \mr + (\gamma - 1) \sum_{e \in E} \delta \ell_ex_e^* \\
&\leq  \delta \sum_{e \in E} d_e x_e^* + \gamma \mr + (\gamma - 1) \sum_{e \in E} \delta d_ex_e^*\\
&=\gamma \delta \sum_{e \in E} d_e x_e^* + \gamma \mr \\
&\leq \gamma \delta [\alpha \opt(\bd) + \beta \mr ] + \gamma \mr \\
&= \gamma\delta\alpha \cdot \opt (\bd) + (\gamma \delta \beta + \gamma) \cdot \mr.\\
\end{align*}
i.e., $\sol$ is $(\gamma \delta \alpha, \gamma \delta \beta + \gamma)$-robust as desired.
\end{proof}
\section{Algorithm for the Robust Traveling Salesman Problem}
\label{section:tsp}

In this section, we give a robust algorithm for the traveling salesman problem:

\begin{theorem}\label{thm:tsp-main}
There exists a $(4.5, 3.75)$-robust algorithm for the traveling salesman problem.
\end{theorem}

Recall that we consider the version of the problem where we are allowed to use edges multiple times in \tsp.
We recall that any \tsp tour must contain a spanning tree, and an Eulerian walk on a doubled \mst is a 2-approximation algorithm for \tsp (known as the ``double-tree algorithm''). One might hope that since we have a $(1, 2)$-robust algorithm for robust \mst, one could take its output and apply the double-tree algorithm to get a $(2, 4)$-robust solution to robust TSP. Unfortunately, we show in Section~\ref{section:doubletree} that this algorithm is not $(\alpha, \beta)$-robust for any $\alpha = o(n), \beta = o(n)$. Nevertheless, we are able to leverage the connection to \mst to arrive at a $(4.5, 3.75)$-robust algorithm for \tsp, given in Section~\ref{section:tsporacle}.

\subsection{Failure of Double-Tree Algorithm}\label{section:doubletree}
The black-box reduction of \cite{KasperskiZ06} for turning exact algorithms into $(1, 2)$-robust algorithms simply uses the exact algorithm to find the optimal solution when all $d_e$ are set to $\frac{\ell_e + u_e}{2}$ and outputs this solution (see \cite{KasperskiZ06} for details on its analysis). We give an example of a robust \tsp instance where applying the double-tree algorithm to the $(1, 2)$-robust \mst generated by this algorithm does not give a robust \tsp solution. Since the doubling of this \mst is a 2-approximation for \tsp when all $d_e$ are set to $\frac{\ell_e + u_e}{2}$, this example will also show that using an approximation algorithm instead of an exact algorithm in the black-box reduction fails to give any reasonable robustness guarantee as stated in Section~\ref{sec:intro}. 

Consider an instance of robust \tsp with vertices $V = \{v', v_1 \ldots v_n\}$, where there is a ``type-1'' edge from $v'$ to $v_i$ with length $1 - \epsilon$ and where there is a ``type-2'' edge from $v_i$ to $v_{i+1}$ for all $i$, as well as from $v_n$ to $v_1$, with length in the range $[0, 2 - \frac{1}{n-1}]$. 

Consider $\mrs$, which uses $n-1$ type-2 edges and two type-1 edges to connect $v'$ to the rest of the tour\footnote{We do not prove this is $\mrs$: Even if it is not, it suffices to upper bound $\mr$ by this solution's regret.}. Its regret is maximized in the realization where all the type-2 edges it is using have length $2 - \frac{1}{n-1}$ and the type-2 edge it is not using has length 0. In this realization, the optimum solution uses the type 2-edge with length 0, the two type-1 edges adjacent to this type-2 edge once, and then the other $n-2$ type-1 edges twice. So $\mrs$ has cost $(n-1)(2 - \frac{1}{n-1})+2(1-\epsilon) \leq 2(n-1)$ whereas $\opt$ has cost $2(n-1)(1-\epsilon)$. Then, the regret of this solution is at most $n\epsilon$.

When all edge costs are set to $\frac{\ell_e + u_e}{2}$, as long as $\epsilon > \frac{1}{2(n-1)}$ the minimum spanning tree of the graph is a star centered at $v'$, i.e., all the length $1-\epsilon$ edges. So the $(1, 2)$-approximation algorithm outputs this tree for \mst. Doubling this tree gives a solution to the robust \tsp instance that costs $2n(1-\epsilon)$ in all realizations of demands. 

Consider the realization $\bd$ where all type-2 edges have length 0. The minimum regret solution costs $2-2\epsilon$ and is also the optimal solution. If the double-tree solution is $(\alpha, \beta)$-robust we get that:

$$2n(1-\epsilon) \leq \alpha \cdot \opt(\bd) + \beta \cdot \mr \leq \alpha \cdot (2 - 2\epsilon) + \beta n\epsilon.$$ 

Setting $\epsilon$ to e.g. $1/n$ gives that one of $\alpha, \beta$ is $\Omega(n)$. 

\begin{figure}
\begin{center}
\textbf{Minimize } $r$ \textbf{subject to}
\end{center}
\begin{equation}\label{eq:rrtsp}
\begin{array}{lll}
& \forall \emptyset \neq S \subset V:  & \sum_{u \in S, v \in V \backslash S} y_{uv} \geq 2 \\
& \forall u \in V: & \sum_{v \neq u} y_{uv} = 2 \\
&\forall \emptyset \neq S \subset V, u \in S, v \in V \backslash S:& \sum_{e \in \delta(S)} x_{e, u, v} \geq y_{uv}\\
&\forall \bd:& \sum_{e \in E} d_e x_e \leq \opt(\bd) + r\\
&\forall u,v \in V, u \neq v:& 0 \leq y_{uv} \leq 1\\
&\forall e \in E, u, v \in V, v \neq u:& 0 \leq x_{e, u, v} \leq 1\\
&\forall e \in E:& x_e \leq 2
\end{array}
\end{equation}
\caption{The Robust TSP Polytope}
\label{fig:tsplp}
\end{figure}

\subsection{LP Relaxation}

We use the LP relaxation of robust traveling salesman in Figure~\ref{fig:tsplp}. This is the standard subtour LP (see e.g. \cite{Vygen_newapproximation}), but augmented with variables specifying the edges used to visit each new vertex, as well as with the regret constraint set. Integrally, $y_{uv}$ is 1 if splitting the tour into subpaths at each point where a vertex is visited for the first time, there is a subpath from $u$ to $v$ (or vice-versa). That is, $y_{uv}$ is 1 if between the first time $u$ is visited and the first time $v$ is visited, the tour only goes through vertices that were already visited before visiting $u$. $x_{e,u,v}$ is 1 if on this subpath, the edge $e$ is used. We use $x_e$ to denote $\sum_{u, v \in V} x_{e,u,v}$ for brevity. We discuss in this subsection why the constraints other than the regret constraint set in \eqref{eq:rrtsp} are identical to the standard \tsp polytope. This discussion may be skipped without affecting the readability of the rest of the paper.

The standard LP for \tsp is the \textit{subtour LP} (see e.g. \cite{Vygen_newapproximation}), which is as follows:

\begin{equation}\label{eq:subtourlp}
\begin{array}{lll}
 \min& \sum_{(u, v) \in E} c_{uv} y_{uv}&\\
\text{s.t. }& \forall \emptyset \neq S \subset V: & \sum_{(u, v) \in \delta(S)} y_{uv} \geq 2 \\
& \forall u \in V: & \sum_{(u, v) \in E} y_{uv} = 2 \\
&\forall (u, v) \in E:& 0 \leq y_{uv} \leq 1
\end{array}
\end{equation}
where $\delta(S)$ denotes the set of edges with one endpoint in $S$.  Note that because the graph is undirected, the order of $u$ and $v$ in terms such as $(u,v)$, $c_{uv}$, and $y_{uv}$ is immaterial, e.g., there is no distinction between edge $(u,v)$ and edge $(v,u)$ and $y_{uv}$ and $y_{vu}$ denote the same variable. This LP is written for the problem formulation where the triangle inequality holds, and thus we only need to consider tours that are cycles that visit every vertex exactly once. We are concerned, however, with the formulation where the triangle inequality does not necessarily hold, but tours can revisit vertices and edges multiple times. To modify the subtour LP to account for this formulation, we instead let $y_{uv}$  be an indicator variable for whether our solution connects $u$ to $v$ using some path in the graph. Using this definition for $y_{uv}$, the subtour LP constraints then tell us that we must buy a set of paths such that a set of edges directly connecting the endpoints of the paths would form a cycle visiting every vertex exactly once. Then, we introduce variables $x_{e, u, v}$ denoting that we are using the edge $e$ on the path from $u$ to $v$. For ease of notation, we let $x_e = \sum_{u, v \in V} x_{e, u, v}$ denote the number of times a fractional solution uses the edge $e$ in paths. We can then use standard constraints from the canonical shortest path LP to ensure that in an integer solution $y_{uv}$ is set to 1 only if for some path from $u$ to $v$, all edges $e$ on the path have $x_{e, u, v}$ set to 1. 

Lastly, note that the optimal tour does not use an edge more than twice. Suppose a tour uses the edge $e = (u, v)$ thrice. By fixing a start/end vertex for the tour, we can split the tour into $e, P_1, e, P_2, e, P_3$ where $P_1$ is the part of the tour between the first and second use of $e$, $P_2$ is the part of the tour between the second and third use of $e$, and $P_3$ is the part of the tour after the third use of $e$. Because the tour starts and ends at the same vertex ($u$ or $v$), and each of the three uses of edge $e$ goes from $u$ to $v$ or vice versa, the number of $P_1$, $P_2$, and $P_3$ that go from $u$ to $v$ or vice-versa (as opposed to going from $u$ to $u$ or $v$ to $v$) must be odd, and hence not zero. Without loss of generality, we can assume $P_1$ goes from $u$ to $v$. Then, the tour $\overline{P_1}, P_2, e, P_3$, where $\overline{P_1}$ denotes the reversal of $P_1$, is a valid tour and costs strictly less than the original tour. So any tour using an edge more than twice is not optimal. This lets us add the constraint $x_e \leq 2$ to the LP without affecting the optimal solution.

This gives the formulation for \tsp without triangle inequality but with repeated edges allowed:

\begin{equation}\label{eq:subtourlp2}
\begin{array}{lll}
 \min &\sum_{e \in E} c_e x_{e}&\\
\text{s.t. }
& \forall \emptyset \neq S \subset V: & \sum_{u \in S,v \in V \backslash S} y_{uv} \geq 2 \\
& \forall u \in V: & \sum_{v \neq u} y_{uv} = 2 \\
&\forall  \emptyset \neq S \subset V, u \in S, v \in V \backslash S:& \sum_{e \in \delta(S)} x_{e, u, v} \geq y_{uv}\\
&\forall u,v \in V, v \neq u:& 0 \leq y_{uv} \leq 1\\
&\forall e \in E, u, v \in V, v \neq u:& 0 \leq x_{e, u, v} \leq 1\\
&\forall e \in E:& x_e \leq 2
\end{array}
\end{equation}

By integrality of the shortest path polytope, if we let $p_{uv}$ denote the length of the shortest path from $u$ to $v$, then $\sum_{e \in E, u, v \in V} c_e x_{e, u, v} \geq \sum_{u, v \in V}p_{uv}y_{uv}$.  In particular, if we fix the value of $y_{uv}$ the optimal setting of $x_{e, u, v}$ values is to set $x_{e, u, v}$ to $y_{uv}$ for every $e$ on the shortest path from $u$ to $v$. So  \eqref{eq:subtourlp2} without the triangle inequality assumption is equivalent to \eqref{eq:subtourlp} with the triangle inequality assumption. In particular, the integrality gap of \eqref{eq:subtourlp2} is the same as the integrality gap of \eqref{eq:subtourlp}, which is known to be at most 3/2 \cite{Wolsey80}. Then, adding a variable $r$ for the fractional solution's regret and the regret constraint set gives \eqref{eq:rrtsp}.

\subsection{Approximate Separation Oracle}\label{section:tsporacle}
We now describe the separation oracle \textsc{RRTSP-Oracle} used to separate \eqref{eq:rrtsp}. All constraints except the regret constraint set can be separated in polynomial time by solving a min-cut problem. Recall that exactly separating the regret constraint set involves finding an ``adversary'' $\sol$ that maximizes $\max_{\bd} [\sum_{e \in E} d_ex_e - \sol(\bd)]$, and seeing if this quantity exceeds $r$. However, since TSP is \np-hard, finding this solution in general is \np-hard. Instead, we will only consider a solution $\sol$ if it is a walk on some spanning tree $T$, and find the one that maximizes $\max_{\bd} [\sum_{e \in E} d_ex_e - \sol(\bd)]$. 

Fix any $\sol$ that is a walk on some spanning tree $T$. For any $e$, if $e$ is not in $T$, the regret of $\bx, \by$ against $\sol$ is maximized by setting $e$'s length to $u_e$. If $e$ is in $T$, then $\sol$ is paying $2d_e$ for that edge whereas the fractional solution pays $d_ex_e \leq 2d_e$, so to maximize the fractional solution's regret, $d_e$ should be set to $\ell_e$. This gives that the regret of fractional solution $\bx$ against any $\sol$ that is a spanning tree walk on $T$ is

$$\sum_{e \in T}(\ell_ex_e - 2\ell_e) + \sum_{e \notin T} u_ex_e = \sum_{e \in E} u_ex_e - \sum_{e \in T}(u_ex_e - (\ell_ex_e - 2\ell_e)).$$

The quantity $\sum_{e \in E} u_ex_e$ is fixed with respect to $T$, so finding the spanning tree $T$ that maximizes this quantity is equivalent to finding $T$ that minimizes $\sum_{e \in T}(u_ex_e - (\ell_ex_e - 2\ell_e))$. But this is just an instance of the minimum spanning tree problem where edge $e$ has weight $u_ex_e - (\ell_ex_e - 2\ell_e)$, and thus we can find $T$ in polynomial time. This gives the following lemma:

\begin{lemma}\label{lemma:tspsep}
For any instance of robust traveling salesman there exists an algorithm \textsc{RRTSP-Oracle} that given a solution $(\bx, \by, r)$ to \eqref{eq:rrtsp} either:
\begin{itemize}
\item Outputs a separating hyperplane for \eqref{eq:rrtsp}, or
\item Outputs ``Feasible'', in which case $(\bx, \by)$ is feasible for the (non-robust) TSP LP and $\forall \bd: \sum_{e \in E} d_e x_e \leq 2 \cdot \opt (\bd) + r$.
\end{itemize}
\end{lemma}

\begin{figure}
\fbox{\begin{minipage}{.98\textwidth}
\textsc{RRTSP-Oracle}($G(V, E)$, $\{(\ell_e, u_e)\}_{e \in E}$, $(\bx, \by, r)$)

\begin{algorithm}[H]
\KwData{Undirected graph $G(V, E)$, lower and upper bounds on edge lengths $\{(\ell_e, u_e)\}_{e \in E}$, solution ($\bx = \{x_{e, u, v}\}_{e \in E, u, v \in V},\by = \{y_{uv}\}_{u, v \in V}, r$) to \eqref{eq:rrtsp}}
Check all constraints of \eqref{eq:rrtsp}, \textbf{return} any violated constraint that is found\;
$G' \leftarrow $ copy of $G$ where $e$ has weight $u_ex_e - (\ell_ex_e - 2\ell_e)$\;
$T' \leftarrow $ minimum spanning tree of $G'$\;
\eIf{$\sum_{e \in T'}(\ell_ex_e - 2\ell_e) + \sum_{e \notin T'} u_ex_e > r$} {
\textbf{return} $\sum_{e \in T'}(\ell_ex_e - 2\ell_e) + \sum_{e \notin T'} u_ex_e \leq r$\;
}{ 
\textbf{return} ``Feasible''
}
\end{algorithm}
\end{minipage}}
\caption{Separation Oracle for \eqref{eq:rrtsp}.}
\label{fig:tsporacle}
\end{figure}

\begin{proof}[Proof of Lemma~\ref{lemma:tspsep}]
\textsc{RRTSP-Oracle} is given in Figure~\ref{fig:tsporacle}. All inequalities except the regret constraint set can be checked exactly by \textsc{RRTSP-Oracle}. Consider the tree $T'$ computed in \textsc{RRTSP-Oracle} and $\bd'$ with  $d_e' = \ell_e$ for $e \in T'$ and $d_e' = u_e$ for $e \notin T'$. The only other violated inequality \textsc{RRTSP-Oracle} can output is the inequality $\sum_{e \in T'}(\ell_ex_e - 2\ell_e) + \sum_{e \notin T'} u_ex_e \leq r$ in line 5, which is equivalent to $\sum_{e \in E} d_e' x_e \leq 2\sum_{e \in T} d_e' + r$. Since $2\sum_{e \in T} d_e'$ is the cost of a tour in realization $\bd'$ (the tour that follows a DFS on the spanning tree $T$), this inequality is implied by the inequality $\sum_{e \in E} d_e' x_e \leq \opt(\bd') + r$ from the regret constraint set. Furthermore, \textsc{RRTSP-Oracle} only outputs this inequality when it is actually violated. 

So, it suffices to show that if there exists $\bd$ such that $\sum_{e \in E} d_e x_e > 2 \opt(\bd) + r$ then \textsc{RRTSP-Oracle} outputs a violated inequality on line 5. Since $\opt(\bd)$ visits all vertices, it contains some spanning tree $T$, such that $\opt(\bd) \geq \sum_{e \in T} d_e$. Combining these inequalities gives 

$$\sum_{e \in E}d_ex_e > 2 \sum_{e \in T}d_e + r.$$

Since all $x_e$ are at most 1, setting $d_e = \ell_e$ for $e \in T$ and $d_e = u_e$ otherwise can only increase $\sum_{e \in E}d_ex_e - 2 \sum_{e \in T}d_e$, so

$$\sum_{e \in T}\ell_ex_e + \sum_{e \notin T}u_ex_e > 2 \sum_{e \in T}\ell_e + r \implies \sum_{e \in E}u_ex_e  - \sum_{e \in T}(u_ex_e - (\ell_ex_e - 2\ell_e)) > r.$$

Then \textsc{RRTSP-Oracle} finds a minimum spanning tree $T'$ on $G'$, i.e. $T'$ such that

$$\sum_{e \in T'}(u_ex_e - (\ell_ex_e - 2\ell_e)) \leq \sum_{e \in T}(u_ex_e - (\ell_ex_e - 2\ell_e)).$$

which combined with the previous inequality gives

$$\sum_{e \in E}u_ex_e  - \sum_{e \in T'}(u_ex_e - (\ell_ex_e - 2\ell_e)) > r \implies \sum_{e \in T'}(\ell_ex_e - 2\ell_e) + \sum_{e \notin T'}u_ex_e > r.$$
\end{proof}

By using the ellipsoid method with separation oracle \textsc{RRTSP-Oracle} and the fact that \eqref{eq:rrtsp} has optimum at most \mr, we get a $(2, 1)$-robust fractional solution. Applying Theorem~\ref{thm:generalrounding} as well as the fact that the TSP polytope has integrality gap $3/2$ (see e.g. \cite{Vygen_newapproximation}) and the TSP problem has a $3/2$-approximation gives Theorem~\ref{thm:tsp-main}.

\section{A Local Search Algorithm for Steiner Tree}\label{section:localsearch}
In this section, we describe a local search algorithm for the Steiner tree, given in \cite{GGKMSSV17}. We show that the algorithm is $4$-approximate, but as with many local search algorithms could run in superpolynomial time in the worst case. Standard tricks can be used to modify this algorithm into a polynomial time $(4+\epsilon)$-approximation. This algorithm will serve as a primitive in the algorithms we design in Section~\ref{section:steiner}.

The local moves considered by the algorithm are all \textit{path swaps}, defined follows: If the current Steiner tree is $T$, the algorithm can pick any two vertices $u, v$ in $T$ such that there exists a path from $u$ to $v$ where all vertices except $u$ and $v$ on the path are not in $T$ (and thus all edges on the path are not in $T$). The algorithm may take any path $f$ from $u$ to $v$ of this form.  It suffices to consider only the shortest path of this form.  The path $f$ is added to $T$, inducing a cycle.  The algorithm then picks a subpath $a$ in the cycle and removes it from $T$, while maintaining that $T$ is a feasible Steiner tree.  It suffices to consider only maximal subpaths. These are just the subpaths formed by splitting the cycle at every vertex with degree at least 3 in $T \cup \{f\}$.  Let $n$ denote the number of nodes and $m$ the number of edges in the graph.  Since there are at most $\genfrac{(}{)}{0pt}{1}{n}{2}$ pairs of vertices $u, v$, and a shortest path $f$ between $u$ and $v$ can be found in $O(m + n\log n)$ time, and all maximal subpaths in the induced cycle in $T \cup \{f\}$ can be found in $O(n)$ time, if there is a move that improves the cost of $T$, we can find it in polynomial time. 

We will use the following observation to show the approximation ratio.

\begin{observation}\label{obs:treevertices}
For any tree the fraction of vertices with degree at most 2 is strictly greater than $1/2$.
\end{observation}
\begin{proof}
This follows from a Markov bound on the random variable $X$ defined as the degree of a uniformly random vertex minus one. A tree with $n$ vertices has average degree $\frac{2n-2}{n} < 2$, so $\mathbb{E}[X] < 1$. In turn, the fraction of vertices with degree 3 or greater is $\Pr[X \geq 2] < 1/2$.
\end{proof}

\begin{theorem}\label{thm:localsearch}
Let $A$ be a solution to an instance of Steiner tree such that no path-swap reduces the cost of $A$. Then $A$ is a 4-approximation.
\end{theorem}
\begin{proof}
Consider any other solution $F$ to the Steiner tree instance. 
We partition the edges of $A$ into the subpaths such that these subpaths' endpoints are (i) vertices with degree 3 or larger in $A$, or (ii) vertices in and $A$ and $F$ (which might also have degree 3 or larger in $A$).  Besides the endpoints, all other vertices in each subpath have degree 2 and are in $A$ but not in $F$. Note that a vertex may appear as the endpoint of more than one subpath. Note also that the set of vertices in $F$ includes the terminals, which, without loss of generality includes all leaves in $A$. This along with condition (i) for endpoints ensures the partition into subpaths is well-defined, i.e., if a subpath ends at a leaf of $A$, that leaf is in $F$.  

We also decompose $F$ into subpaths, but some edges may be contained in two of these subpaths. To decompose $F$ into subpaths, we first partition the edges of $F$ into maximal connected subgraphs of $F$ whose leaves are vertices in $A$ (including terminals) and whose internal vertices are not in $A$. Note that some vertices may appear in more than one subgraph, e.g., an internal vertex in $F$ that is in $A$ becomes a leaf in multiple subgraphs.  Since these subgraphs are not necessarily paths, we next take any DFS walk on each of these subgraphs starting from one of their leaves (that is, one of the vertices in $A$). We take the route traversed by the DFS walk, and split it into subpaths at every point where the DFS walk reaches a leaf. This now gives a set of subpaths in $F$ such that each subpaths' endpoints are vertices in $A$, no other vertices on a subpath are in $A$, and no edge appears in more than two subpaths. Let $\mathcal{A}$ and $\mathcal{F}$ denote the set of subpaths we decomposed $A$ and $F$ into, respectfully.

For $a \in \mathcal{A}$ let $N(a) \subseteq \mathcal{F}$ be the set of subpaths $f$ in $\mathcal{F}$ such that $A \setminus a \cup f$ is a feasible Steiner tree, i.e. $f$ can be swapped for $a$, and let $N(X) = \cup_{a \in X} N(a)$. We will show that for any $X \subseteq \mathcal{A}$, $|N(X)| \geq \frac{1}{2}|X|$. By an extension of Hall's Theorem (Fact 15 in \cite{GGKMSSV17}) this implies the existence of a weight function $\alpha: \mathcal{A} \times \mathcal{F} \rightarrow \mathbb{R}^+$ such that:
\begin{enumerate}
    \item $\alpha(a, f) > 0$ only if $f$ can be swapped for $a$
    \item For all subpaths $a \in \mathcal{A}$, $\sum_{f \in N(a)} \alpha(a, f) = 1$.
    \item For all subpaths $f \in \mathcal{F}$, $\sum_{a \in N^{-1}(f)} \alpha(a, f) \leq 2$.
\end{enumerate}  

This weight function then gives:

$$c(A) = \sum_{a \in \mathcal{A}} c(a) = \sum_{a \in \mathcal{A}} \sum_{f \in N(a)} c(a) \alpha(a, f) \leq \sum_{f \in \mathcal{F}} \sum_{a \in N^{-1}(f)} c(f) \alpha(a, f) \leq \sum_{f \in \mathcal{F}} 2c(f) \leq 4c(F),$$
where $N^{-1}(f) = \{a \in \mathcal{A}: f \in N(a)\}$.  The first inequality holds by the assumption in the lemma statement that no swaps reduce the cost of $A$, so for $a \in \mathcal{A}$ and $f \in N(a)$, $c(a) \leq c(f)$. The last inequality follows by the fact that every edge in $F$ appears in at most two subpaths in $\mathcal{F}$.

We now turn towards proving that for any $X \subseteq \mathcal{A}$, $|N(X)| \geq \frac{1}{2}|X|$. Fix any $X \subseteq \mathcal{A}$.   Suppose that we remove all of the edges on paths $a \in X$ from $A$, and also remove all vertices on these paths except their endpoints.  After removing these nodes and edges, we are left with $|X|+1$ connected components.  Let $T'$ be a tree with $|X|+1$ vertices, one for each of these connected components, with an edge between any pair of vertices in $T'$ whose corresponding components are connected by a subpath $a \in X$. Consider any vertices of $T'$ with degree at most 2. We claim the corresponding component contains a vertex in $F$. Let $V'$ denote the set of vertices in the corresponding component that are endpoints of subpaths in $\mathcal{A}$. There must be at least one such vertex in $V'$.  Furthermore, it is not possible that all of the vertices in $V'$ are internal vertices of $A$ with degree at least 3, since at most two subpaths leave this component and there are no cycles in $A$. The only other option for endpoints is vertices in $F$, so this component must contain some vertex in $F$.

Applying Observation~\ref{obs:treevertices}, strictly more than $(|X|+1)/2$ (i.e., at least $|X|/2 + 1$) of the components have degree at most 2, and by the previous argument contain a vertex in $F$. These vertices are connected by $F$, and since each subpath in $\mathcal{F}$ does not have internal vertices in $A$, no subpath in $\mathcal{F}$ passes through more than two of these components. Hence, at least $|X|/2$ of the subpaths in $\mathcal{F}$ have endpoints in two different components because at least $|X|/2$ edges are required to connect $|X|/2+1$ vertices. In turn, any of these $|X|/2$ paths $f$ could be swapped for one of the subpaths $a \in X$ that is on the path between the components containing $f$'s endpoints. This shows that $|N(X)| \geq |X|/2$ as desired.
\end{proof}
\section{Algorithm for the Robust Steiner Tree Problem}\label{section:steiner}

\begin{figure}
\begin{center}
 \textbf{Minimize } $r$ \textbf{subject to}
\begin{eqnarray} 
\label{eq:steiner}\forall S \subset V \text{ such that } \emptyset \subset S \cap T \subset T: &\sum_{e \in \delta(S)} x_e \geq 1\\
\label{eq:range}\forall \bd \text{ such that } d_e \in [\ell_e, u_e]: & \sum_{e \in E} d_e x_e \leq \opt(\bd) + r\\
\forall e \in E: & x_e \in [0, 1] 
\end{eqnarray}
\end{center}
\caption{The Robust Steiner Tree Polytope}
\label{fig:rrst}
\end{figure}

In this section, our goal is to find a fractional solution to the LP in Fig.~\ref{fig:rrst} for robust Steiner tree. By Theorem~\ref{thm:generalrounding} and known approximation/integrality gap results for Steiner Tree, this will give the following theorem:

\begin{theorem}\label{thm:steiner-main}
There exists a $(2755, 64)$-robust algorithm for the Steiner tree problem.
\end{theorem}

It is well-known that the standard Steiner tree polytope admits an exact separation oracle
(by solving the $s, t$-min-cut problem using edge weights $x_e$ for all $s, t \in T$) so it is sufficient to find an approximate separation oracle for the regret constraint set. We present the main ideas of the separation oracle here, and defer the details to Section~\ref{section:differenceapprox2}.

First, we create the derived instance of the Steiner tree problem which is a copy $G'$ of the input graph $G$ with edge weights $u_e x_e +\ell_e -\ell_e x_e$. As noted earlier, the optimal Steiner tree $T^*$ on the derived instance maximizes the regret of the fractional solution $\bx$. However, since Steiner tree is \np-hard, we cannot hope to exactly find $T^*$. We need a Steiner tree $\hat{T}$ such that the regret caused by it can be bounded against that caused by $T^*$. The difficulty is that the regret corresponds to the total weight of edges \textit{not} in the Steiner tree (plus an offset that we will address later), whereas standard Steiner tree approximations give guarantees in terms of the total weight of edges in the Steiner tree. 
%
We overcome this difficulty by requiring a stricter notion of ``difference approximation'' -- that the weight of edges $\hat{T} \setminus T^*$ be bounded against those in $T^* \setminus \hat{T}$. Note that this condition ensures that not only is the weight of edges in $\hat{T}$ bounded against those in $T^*$, but also that the weight of edges {\em not in $\hat{T}$} is bounded against that of edges {\em not in $T^*$}. We show the following lemma to obtain the difference approximation:

\begin{lemma}\label{lemma:differenceapprox}
For any $\epsilon > 0$, there exists a polynomial-time algorithm for the Steiner tree problem such that if $\opt$ denotes the set of edges in the optimal solution and $c(S)$ denotes the total weight of edges in $S$, then for any input instance of Steiner tree, the output solution $\alg$ satisfies $c(\alg \setminus \opt) \leq (4 + \epsilon) \cdot c(\opt \setminus \alg)$.
\end{lemma}

\begin{proof}[Proof]
The algorithm we use is the local search algorithm described in Section~\ref{section:localsearch}, which provides a 4-approximation $\alg$, i.e., $c(\alg) \leq 4 \cdot c(\opt)$.  Suppose that the cost of each edge $e \in \alg \cap \opt$ is now changed from its initial value to 0.  After this change, $\alg$ remains locally optimal because for every feasible solution $F$ that can be reached by making a local move from $\alg$, the amount by which the cost of $\alg$ has decreased by setting edge costs to zero is at least the amount by which $F$ has decreased.  Hence no local move causes a decrease in cost.  Thus, $\alg$ remains a 4-approximation, which implies that $c(\alg \setminus \opt) \leq 4 \cdot c(\opt \setminus \alg)$. 

We also need to show that the algorithm converges in polynomially many iterations. The authors in \cite{GGKMSSV17} achieve this convergence by discretizing all the edge costs to the nearest multiple of $\frac{\epsilon}{kn} c(\apx)$ for an initial solution $\apx$ such that $c(\opt) \leq c(\apx) \leq k c(\opt)$ (e.g., a simple way to do so is to start with a solution formed by the union of shortest paths between terminals, and then remove edges which cause cycles arbitrarily. This solution has cost between $c(\opt)$ and $n^2 \cdot c(\opt)$. See Section B.3 of \cite{GGKMSSV17} for more details). This guarantees that the algorithm converges in $\frac{kn}{\epsilon}$ iterations, at an additive $\epsilon \cdot c(\opt)$ cost. For a standard approximation algorithm this is not an issue, but for an algorithm that aims for a guarantee of the form $c(\alg \setminus \opt) \leq O(1) \cdot c(\opt \setminus \alg)$ an additive $\epsilon \cdot c(\opt)$ might be too much. 

We modify the algorithm as follows to ensure that it converges in polynomially many iterations: We only consider swapping out $a$ for $f$ if the decrease in cost is at least $\epsilon/4$ times the cost of $a$, and we always choose the swap of this kind that decreases the cost by the largest amount\footnote{Note that $c(a) - c(f)$ and $\frac{c(a) - c(f)}{c(a)}$ are both maximized by maximizing $c(a)$ and minimizing $c(f)$. Any path $f$ from $u$ to $v$ that we consider adding is independent of the paths $a$ we can consider removing, since $f$ by definition does not intersect with our solution. So in finding a swap satisfying these conditions if one exists, it still suffices to only consider swaps between shortest paths $f$ and longest paths $a$ in the resulting cycles as before.}. 

We now show the algorithm converges. Later in Section~\ref{section:differenceapprox2}, we will prove two claims, so for brevity's sake we will not include the proofs of the claims here. The first claim is that as long as $c(\alg \setminus \opt) > (4+\epsilon)c(\opt \setminus \alg)$, there is a swap between $\alg$ and $\opt$ where decrease in cost is at least $\epsilon/4$ times the cost of the path being swapped out, and is at least $\epsilon/4n^2$ times $c(\alg \setminus \opt)$ (the proof follows similarly to Lemma~\ref{lemma:greedyswap} in Section~\ref{section:differenceapprox2}). The second claim is that in any swap the quantity $c(\alg \setminus \opt) - c(\opt \setminus \alg)$ decreases by the same amount that $c(\alg)$ does (see Observation~\ref{obs:swap} in Section~\ref{section:differenceapprox2}). 

So, we use $c(\alg \setminus \opt) - c(\opt \setminus \alg)$ as a potential to bound the number of swaps. This potential is initially at most $n \max_e c_e$, is always at least $\min_e c_e$ as long as $c(\alg \setminus \opt) > c(\opt \setminus \alg)$, and each swap decreases it multiplicatively by at least a factor of $(1-\epsilon/4n^2)$ as long as $c(\alg \setminus \opt) > (4+\epsilon)c(\opt \setminus \alg)$. Thus the algorithm only needs to make $\frac{\log (n \max_e c_e / \min_e c_e)}{-\log (1-\epsilon/4n^2)}$ swaps to arrive at a solution that is a $(4+\epsilon)$-approximation, which is a polynomial in the input size.
\end{proof}

\subsection{Special Case where the Lower Bounds on All Edge Lengths Are $\ell_e = 0$}

Recall that the regret caused by $T$ is not exactly the weight of edges not in $T$, but includes a fixed offset of $\sum_{e \in E} (\ell_e - \ell_e x_e)$. If $\ell_e = 0$ for all edges, i.e., the offset is $0$, then we can recover a robust algorithm from Lemma~\ref{lemma:differenceapprox} alone with much better constants than in Theorem~\ref{thm:steiner-main}:

\begin{lemma}\label{lemma:steinersep1}
For any instance of robust Steiner tree for which all $\ell_e = 0$, for every $\epsilon > 0$ there exists an algorithm \textsc{RRST-Oracle-ZLB} which, given a solution $(\bx, r)$ to the LP in Fig.~\ref{fig:rrst}, either:
\begin{itemize}
\item Outputs a separating hyperplane for the LP in Fig.~\ref{fig:rrst}, or
\item Outputs ``Feasible'', in which case $\bx$ is feasible for the (non-robust) Steiner tree LP and $\forall \bd: \sum_{e \in E} d_e x_e \leq \opt (\bd) + (4 + \epsilon)r$.
\end{itemize}
\end{lemma}

\begin{figure}
\fbox{\begin{minipage}{.98\textwidth}
\textsc{RRST-Oracle-ZLB}($G(V, E)$, $\{u_e\}_{e \in E}$, $(\bx, r)$)

\begin{algorithm}[H]

\KwData{Undirected graph $G(V, E)$, upper bounds on edge lengths $\{u_e\}_{e \in E}$, solution ($\bx = \{x_e\}_{e \in E}, r$) to the LP in Fig.~\ref{fig:rrst}}
Check all constraints of the LP in Fig.~\ref{fig:rrst} except the regret constraint set, \textbf{return} any violated constraint that is found\;
$G' \leftarrow $ copy of $G$ where $e$ has cost $u_ex_e$ \;
$T' \leftarrow $ output of algorithm from Lemma~\ref{lemma:differenceapprox} on $G'$\;
\eIf{$\sum_{e \notin T'} u_ex_e > r$}{
\textbf{return} $\sum_{e \notin T'} u_ex_e \leq r$\;
}{
\textbf{return} ``Feasible''\;
}
\end{algorithm}
\end{minipage}}
\caption{Separation Oracle for the LP in Fig.~\ref{fig:rrst} when $\ell_e = 0,~\forall e$}
\label{fig:storacle}
\end{figure}

\textsc{RRST-Oracle-ZLB} is given in Fig.~\ref{fig:storacle}. Via the ellipsoid method this gives a $(1, 4+\epsilon)$-robust fractional solution. Using Theorem~\ref{thm:generalrounding}, the fact that the integrality gap of the LP we use is 2 \cite{Vazirani01}, and that there is a $(\ln 4+\epsilon) \approx 1.39$-approximation for Steiner tree \cite{ByrkaGRS10}, with appropriate choice of $\epsilon$ we get the following corollary:

\begin{corollary}\label{cor:steineralg1}
There exists a $(2.78, 12.51)$-robust algorithm for Steiner tree when $\ell_e = 0$ for all $e \in E$.
\end{corollary}

\begin{proof}[Proof of Lemma~\ref{lemma:steinersep1}]
All inequalities except the regret constraint set can be checked exactly by \textsc{RRST-Oracle-ZLB}. Consider the tree $T'$ computed in \textsc{RRST-Oracle-ZLB} and $\bd'$ with  $d_e' = 0$ for $e \in T'$ and $d_e' = u_e$ for $e \notin T'$. The only other violated inequality \textsc{RRST-Oracle-ZLB} can output is the inequality $\sum_{e \notin T'} u_e x_e \leq r$ in line 5, which is equivalent to $\sum_{e \in E} d_e' x_e \leq T'(\bd') + r$, an inequality from the regret constraint set. Furthermore, \textsc{RRST-Oracle-ZLB} only outputs this inequality when it is actually violated. So, it suffices to show that if there exists $\bd, \sol$ such that $\sum_{e \in E} d_e x_e > \sol(\bd) + (4 + \epsilon)r$ then \textsc{RRST-Oracle-ZLB} outputs a violated inequality on line 5, i.e., finds Steiner tree $T'$ such that $\sum_{e \notin T'} u_e x_e > r$.

Suppose there exists $\bd, \sol$ such that $\sum_{e \in E} d_e x_e > \sol(\bd) + (4 + \epsilon)r$. Let $\bd^*$ be the vector obtained from $\bd$ by replacing $d_e$ with $u_e$ for edges not in $\sol$ and with $0$ for edges in $\sol$. Replacing $\bd$ with $\bd^*$ can only increase $\sum_{e \in E} d_e x_e - \sol(\bd)$, i.e.:

\begin{equation}\label{eq:dstarreg}
\sum_{e \notin \sol} u_e x_e = \sum_{e \in E} d_e^* x_e > \sol(\bd^*) + (4 + \epsilon)r = (4 + \epsilon)r.
\end{equation}

Consider the graph $G'$ made by $\textsc{RRST-Oracle-ZLB}$. We'll partition the edges into four sets, $E_0, E_S, E_T, E_{ST}$ where $E_0$ = $E \setminus (\sol \cup T')$, $E_S = \sol \setminus T'$, $E_T = T' \setminus \sol$, $E_{ST} = \sol \cap T'$. Let $c(E')$ for $E' = E_0, E_S, E_T, E_{ST}$ denote $\sum_{e \in E'} u_ex_e$, i.e., the total cost of the edge set $E'$ in $G'$. Since $\bd^*$ has $d_e = 0$ for $e \in \sol$, from \eqref{eq:dstarreg} we get that $c(E_0) + c(E_{T}) > (4 + \epsilon) r$.

Now note that $\sum_{e \notin T'} u_e x_e = c(E_0) + c(E_{S})$. Lemma~\ref{lemma:differenceapprox} gives that $(4 + \epsilon)c(E_{S}) \geq c(E_{T})$. Putting it all together, we get that:

$$\sum_{e \notin T'} u_e x_e = c(E_0) + c(E_S) \geq c(E_0) + \frac{c(E_T)}{4 + \epsilon} \geq \frac{c(E_0) + c(E_T)}{4 + \epsilon} > \frac{(4+\epsilon)r}{4 + \epsilon} = r.$$
\qedhere
\end{proof}

\subsection{General Case for Arbitrary Lower Bounds on Edge Lengths}

In general, the approximation guarantee given in Lemma~\ref{lemma:differenceapprox} alone does not suffice because of the offset of $\sum_{e \in E} (\ell_e - \ell_e x_e)$. We instead rely on a stronger notion of approximation formalized in the next lemma that provides simultaneous approximation guarantees on two sets of edge weights: $c_e = u_e x_e  - \ell_e x_e + \ell_e$ and $c_e' = \ell_e - \ell_e x_e$. The guarantee on $\ell_e - \ell_e x_e$ can then be used to handle the offset.

\begin{lemma}\label{lemma:differenceapprox2}
Let $G$ be a graph with terminals $T$ and two sets of edge weights $c$ and $c'$. Let $\sol$ be any Steiner tree connecting $T$. Let $\Gamma' > 1$, $\kappa > 0$, and $0 < \epsilon < \frac{4}{35}$ be fixed constants.
Then there exists a constant $\Gamma$ (depending on $\Gamma', \kappa, \epsilon$) and an algorithm that obtains a collection of Steiner trees $\alg$, at least one of which (let us call it $\alg_i$) satisfies:
\begin{itemize}
    \item $c(\alg_i \setminus \sol) \leq 4\Gamma \cdot c(\sol \setminus \alg_i),$ and
    \item $c'(\alg_i) \leq (4\Gamma' +\kappa + 1 + \epsilon) \cdot c'(\sol)$.
\end{itemize}
\end{lemma}

The fact that Lemma~\ref{lemma:differenceapprox2} generates multiple solutions (but only polynomially many) is fine because as long as we can show that one of these solutions causes sufficient regret, our separation oracle can just iterate over all solutions until it finds one that causes sufficient regret. 

We give a high level sketch of the proof of Lemma~\ref{lemma:differenceapprox2} here, and defer the full details to Section~\ref{section:differenceapprox2}. The algorithm uses the idea of \textit{alternate minimization}, alternating between a ``forward phase'' and a ``backward phase''. The goal of each forward phase/backward phase pair is to keep $c'(\alg)$ approximately fixed while obtaining a net decrease in $c(\alg)$. In the forward phase, the algorithm greedily uses local search, choosing swaps that decrease $c$ and increase $c'$ at the best ``rate of exchange'' between the two costs (i.e., the maximum ratio of decrease in $c$ to increase in $c'$), until $c'(\alg)$ has increased past some upper threshold. Then, in the backward phase, the algorithm greedily chooses swaps that decrease $c'$ while increasing $c$ at the best rate of exchange, until $c'(\alg)$ reaches some lower threshold, at which point we start a new forward phase. 

We guess the value of $c'(\sol)$ (we can run many instances of this algorithm and generate different solutions based on different guesses for this purpose) and set the upper threshold for $c'(\alg)$ appropriately so that we satisfy the approximation guarantee for $c'$. For $c$ we show that as long as $\alg$ is not a $4\Gamma$-difference approximation with respect to $c$ then a forward/backward phase pair reduces $c(\alg)$ by a non-negligible amount (of course, if $\alg$ is a $4\Gamma$-difference approximation then we are done). This implies that after enough iterations, $\alg$ must be a $4\Gamma$-difference approximation as $c(\alg)$ can only decrease by a bounded amount. To show this, we claim that while $\alg$ is not a $4\Gamma$-difference approximation, for sufficiently large $\Gamma$ the following bounds on rates of exchange hold:

\begin{itemize} 
\item For each swap in the forward phase, the ratio of decrease in $c(\alg)$ to increase in $c'(\alg)$ is at least some constant $k_1$ times $\frac{c(\alg \setminus \sol)}{c'(\sol \setminus \alg)}$.

\item For each swap in the backward phase, the ratio of increase in $c(\alg)$ to decrease in $c'(\alg)$ is at most some constant $k_2$ times $\frac{c(\sol \setminus \alg)}{c'(\alg \setminus \sol)}$. 
\end{itemize} 
Before we discuss how to prove this claim, let us see why this claim implies that a forward phase/backward phase pair results in a net decrease in $c(\alg)$. If this claim holds, suppose we set the lower threshold for $c'(\alg)$ to be, say, $101 c'(\sol)$. That is, throughout the backward phase, we have $c'(\alg) > 101 c'(\sol)$. This lower threshold lets us rewrite our upper bound on the rate of exchange in the backward phase in terms of the lower bound on rate of exchange for the forward phase:

$$k_2 \frac{c(\sol \setminus \alg)}{c'(\alg \setminus \sol)} \leq k_2 \frac{c(\sol \setminus \alg)}{c'(\alg) - c'(\sol)} \leq k_2 \frac{c(\sol \setminus \alg)}{100 c'(\sol)} \leq  k_2 \frac{c(\sol \setminus \alg)}{100 c'(\sol \setminus \alg)} $$
$$\leq k_2 \frac{1}{4\Gamma} \frac{c(\alg \setminus \sol)}{100 c'(\sol \setminus \alg)} = \frac{k_2}{400\Gamma k_1} \cdot k_1\frac{c(\alg \setminus \sol)}{c'(\sol \setminus \alg)}.$$

In other words, the bound in the claim for the rate of exchange in the forward phase is larger than the bound for the backward phase by a multiplicative factor of $\Omega(1) \cdot \Gamma$. While these bounds depend on $\alg$ and thus will change with every swap, let us make the simplifying assumption that through one forward phase/backward phase pair these bounds remain constant. Then, the change in $c(\alg)$ in one phase is just the rate of exchange for that phase times the change in $c'(\alg)$, which by definition of the algorithm is roughly equal in absolute value for the forward and backward phase. So this implies that the decrease in $c(\alg)$ in the forward phase is $\Omega(1) \cdot \Gamma$ times the increase in $c(\alg)$ in the backward phase, i.e., the net change across the phases is a non-negligible decrease in $c(\alg)$ if $\Gamma$ is sufficiently large. Without the simplifying assumption, we can still show that the decrease in $c(\alg)$ in the forward phase is larger than the increase in $c(\alg)$ in the backward phase for large enough $\Gamma$ using a much more technical analysis. In particular, our analysis will show there is a net decrease as long as:
\begin{equation}\label{eq:maincond}
\min\left\{\frac{4\Gamma - 1}{8\Gamma} ,\frac{(4\Gamma - 1)(\sqrt{\Gamma}-1)(\sqrt{\Gamma}-1-\epsilon)\kappa}{16(1+\epsilon)\Gamma^2}\right\} - (e^{\zeta' (4\Gamma'+\kappa+1+\epsilon)} - 1) > 0,
\end{equation}
where 
$$\zeta' = \frac{4(1+\epsilon)\Gamma'}{(\sqrt{\Gamma'}-1)(\sqrt{\Gamma'}-1-\epsilon)(4\Gamma'-1)(4\Gamma-1)}.$$ 
Note that for any positive $\epsilon, \kappa, \Gamma'$, there exists a sufficiently large value of $\Gamma$ for \eqref{eq:maincond} to hold, since as $\Gamma \rightarrow \infty$, we have $\zeta' \rightarrow 0$, so that 
$$(e^{\zeta' (4\Gamma'+\kappa+1+\epsilon)} - 1) \rightarrow 0 \text{ and }$$
$$\min\left\{\frac{4\Gamma - 1}{8\Gamma} ,\frac{(4\Gamma - 1)(\sqrt{\Gamma}-1)(\sqrt{\Gamma}-1-\epsilon)\kappa}{16(1+\epsilon)\Gamma^2}\right\} \rightarrow \min\{1/2, \kappa/(4+4\epsilon)\}.$$ 
So, the same intuition holds: as long as we are willing to lose a large enough $\Gamma$ value, we can make the increase in $c(\alg)$ due to the backward phase negligible compared to the decrease in the forward phase, giving us a net decrease.

It remains to argue that the claimed bounds on rates of exchange hold. Let us argue the claim for $\Gamma = 4$, although the ideas easily generalize to other choices of $\Gamma$. We do this by generalizing the analysis giving Lemma~\ref{lemma:differenceapprox}. This analysis shows that if $\alg$ is a locally optimal solution, then 
$$c(\alg \setminus \sol) \leq 4\cdot c(\sol \setminus \alg),$$
i.e., $\alg$ is a 4-difference approximation of $\sol$. The contrapositive of this statement is that if $\alg$ is not a 4-difference approximation, then there is at least one swap that will improve it by some amount.  We modify the approach of~\cite{GGKMSSV17} by weakening the notion of locally optimal.  In particular, suppose we define a solution $\alg$ to be ``approximately'' locally optimal if at least half of the (weighted) swaps between paths $a$ in $\alg \setminus \sol$ and paths $f$ in $\sol \setminus \alg$ satisfy $c(a) \leq 2c(f)$ (as opposed to $c(a) \leq c(f)$ in a locally optimal solution; the choice of 2 for both constants here implies $\Gamma = 4$). Then a modification of the analysis of the local search algorithm, losing an additional factor of 4, shows that if $\alg$ is approximately locally optimal, then 
$$c(\alg \setminus \sol) \leq 16 \cdot c(\sol \setminus \alg).$$   The contrapositive of this statement, however, has a stronger consequence than before: if $\alg$ is not a 16-difference approximation, then a weighted half of the swaps satisfy $c(a) > 2c(f)$, i.e. reduce $c(\alg)$ by at least 
$$c(a) - c(f) > c(a) - c(a)/2 = c(a)/2.$$
The decrease in $c(\alg)$ due to all of these swaps together is at least $c(\alg \setminus \sol)$ times some constant. In addition, since a swap between $a$ and $f$  increases $c'(\alg)$ by at most $c'(f)$, the total increase in $c'$ due to these swaps is at most $c'(\sol \setminus \alg)$ times some other constant. An averaging argument then gives the rate of exchange bound for the forward phase in the claim, as desired. The rate of exchange bound for the backward phase follows analogously.

This completes the algorithm and proof summary, although more detail is needed to formalize these arguments. Moreover, we also need to show that the algorithm runs in polynomial time. These details are given in Section~\ref{section:differenceapprox2}.

We now formally define our separation oracle \textsc{RRST-Oracle} in Fig.~\ref{fig:storacle2} and prove that it is an approximate separation oracle in the lemma below:

\begin{figure}
\noindent
\fbox{\begin{minipage}{.98\textwidth}
\textsc{RRST-Oracle}($G(V, E)$, $\{[\ell_e, u_e]\}_{e \in E}$, $(\bx, r)$)\newline
\begin{algorithm}[H]
\KwData{Undirected graph $G(V, E)$, lower and upper bounds on edge lengths $\{[\ell_e, u_e]\}_{e \in E}$, solution ($\bx = \{x_e\}_{e \in E}, r$) to the LP in Fig.~\ref{fig:rrst}}
Check all constraints of the LP in Fig.~\ref{fig:rrst} except regret constraint set, \textbf{return} any violated constraint that is found\;
$G' \leftarrow $ copy of $G$ where $c_e = u_ex_e - \ell_ex_e + \ell_e$, $c_e' = \ell_e - \ell_e x_e$\;
$\alg \leftarrow $ output of algorithm from Lemma~\ref{lemma:differenceapprox2} on $G'$\;
\For{$\alg_i \in \alg$}{
\If{$\sum_{e \notin \alg_i} u_ex_e + \sum_{e \in \alg_i} \ell_ex_e - \sum_{e \in \alg_i} \ell_e > r$}{
\textbf{return} $\sum_{e \notin \alg_i} u_ex_e + \sum_{e \in \alg_i} \ell_ex_e - \sum_{e \in \alg_i} \ell_e \leq r$\;
}
}
\textbf{return} ``Feasible''\;
\end{algorithm}
\end{minipage}}
\caption{Separation Oracle for LP in Fig.~\ref{fig:rrst}}
\label{fig:storacle2}
\end{figure}

\begin{lemma}\label{lemma:steinersep2}
Fix any $\Gamma' > 1, \kappa > 0, 0 < \epsilon < 4/35$ and let $\Gamma$ be the constant given in Lemma~\ref{lemma:differenceapprox2}. Let $\alpha = (4\Gamma'+\kappa+2+\epsilon)4\Gamma+1$ and $\beta = 4\Gamma$. Then there exists an algorithm \textsc{RRST-Oracle} that given a solution $(\bx, r)$ to the LP in Fig.~\ref{fig:rrst} either:
\begin{itemize}
\item Outputs a separating hyperplane for the LP in Fig.~\ref{fig:rrst}, or
\item Outputs ``Feasible'', in which case $\bx$ is feasible for the (non-robust) Steiner tree LP and 
$$\forall \bd: \sum_{e \in E} d_e x_e \leq \alpha \cdot \opt (\bd) + \beta \cdot r.$$
\end{itemize}
\end{lemma}
\begin{proof}
It suffices to show that if there exists $\bd, \sol$ such that 
$$\sum_{e \in E} d_e x_e > \alpha \cdot \sol(\bd) + \beta \cdot r, \text { i.e., } \sum_{e \in E} d_ex_e - \alpha \cdot \sol(\bd) > \beta \cdot r$$ 
then \textsc{RRST-Oracle} outputs a violated inequality on line 6, i.e., the algorithm finds a Steiner tree $T'$ such that 
$$\sum_{e \notin T'} u_ex_e + \sum_{e \in T'} \ell_ex_e - \sum_{e \in T'} \ell_e > r.$$ 
Notice that in the inequality 
$$\sum_{e \in E} d_ex_e - \alpha \cdot \sol(\bd) > \beta \cdot r,$$ 
replacing $\bd$ with $\bd'$ where $d_e' = \ell_e$ when $e \in \sol$ and $d_e' = u_e$ when $e \notin \sol$ can only increase the left hand side. So replacing $\bd$ with $\bd'$ and rearranging terms, we have
$$\sum_{e \in \sol} \ell_ex_e + \sum_{e \notin \sol} u_ex_e > \alpha \sum_{e \in \sol} \ell_e + \beta \cdot r = \sum_{e \in \sol} \ell_e + \left[(\alpha-1) \sum_{e \in \sol} \ell_e + \beta \cdot r\right].$$
In particular, the regret of the fractional solution against $\sol$ is at least $(\alpha-1) \sum_{e \in \sol} \ell_e + \beta \cdot r$. 

Let $T'$ be the Steiner tree satisfying the conditions of Lemma~\ref{lemma:differenceapprox2} with $c_e = u_ex_e - \ell_ex_e + \ell_e$ and $c_e' = \ell_e - \ell_e x_e$. Let $E_0$ = $E \setminus (\sol \cup T')$, $E_S = \sol \setminus T'$, and $E_T = T' \setminus \sol$. Let $c(E')$ for $E' = E_0, E_S, E_T$ denote $\sum_{e \in E'} (u_ex_e - \ell_ex_e + \ell_e)$, i.e., the total weight of the edges $E'$ in $G'$. Now, note that the regret of the fractional solution against a solution using edges $E'$ is:

\begin{align*}
\sum_{e \notin E'} u_ex_e + \sum_{e \in E'} \ell_ex_e - \sum_{e \in E'} \ell_e &= \sum_{e \notin E'} (u_ex_e - \ell_ex_e + \ell_e) - \sum_{e \in E}(\ell_e - \ell_ex_e) \\
&= c(E \setminus E') - \sum_{e \in E}(\ell_e - \ell_ex_e).
\end{align*}

Plugging in $E' = \sol$, we then get that:

$$c(E_0) + c(E_T) - \sum_{e \in E}(\ell_e - \ell_ex_e) > (\alpha - 1) \sum_{e \in \sol} \ell_e + \beta \cdot r.$$

Isolating $c(E_T)$ then gives:

\begin{align*}
c(E_T) &> (\alpha - 1) \sum_{e \in \sol} \ell_e + \beta \cdot r - \sum_{e \in E_0} (u_ex_e  - \ell_e x_e + \ell_e) + \sum_{e \in E}(\ell_e - \ell_ex_e) \\
&=(\alpha - 1) \sum_{e \in \sol} \ell_e + \beta \cdot r - \sum_{e \in E_0} u_ex_e + \sum_{e \notin E_0}(\ell_e - \ell_ex_e).
\end{align*}

Since $\beta = 4\Gamma$, Lemma~\ref{lemma:differenceapprox2} along with an appropriate choice of $\epsilon$ gives that $c(E_{T}) \leq \beta c(E_{S})$, and thus:

$$c(E_S) > \frac{1}{\beta} \left[ (\alpha - 1) \sum_{e \in \sol} \ell_e + \beta \cdot r - \sum_{e \in E_0} u_ex_e + \sum_{e \notin E_0}(\ell_e - \ell_ex_e) \right].$$

Recall that our goal is to show that $c(E_0) + c(E_S) - \sum_{e \in E} (\ell_e - \ell_e x_e) > r$, i.e., that the regret of the fractional solution against $T'$ is at least $r$. Adding $c(E_0)- \sum_{e \in E} (\ell_e - \ell_e x_e)$ to both sides of the previous inequality, we can lower bound $c(E_0) + c(E_S) - \sum_{e \in E} (\ell_e - \ell_e x_e)$ as follows:

\begin{align*}
&c(E_0) + c(E_S) - \sum_{e \in E} (\ell_e - \ell_e x_e) \\
&> \frac{1}{\beta} \left[ (\alpha - 1) \sum_{e \in \sol} \ell_e + \beta \cdot r - \sum_{e \in E_0} u_ex_e + \sum_{e \notin E_0}(\ell_e - \ell_ex_e) \right]\\
&\qquad +\sum_{e \in E_0} (u_e x_e  - \ell_e x_e + \ell_e) - \sum_{e \in E} (\ell_e - \ell_e x_e)\\
&= r + \frac{\alpha-1}{\beta} \sum_{e \in \sol} \ell_e + \frac{1}{\beta}\sum_{e \notin E_0}(\ell_e - \ell_ex_e) +\frac{\beta - 1}{\beta} \sum_{e \in E_0} u_e x_e - \sum_{e \notin E_0} (\ell_e - \ell_e x_e)\\
&\geq r + \frac{\alpha-1-\beta}{\beta} \sum_{e \in \sol} \ell_e + \frac{1}{\beta}\sum_{e \notin E_0}(\ell_e - \ell_ex_e) +\frac{\beta - 1}{\beta} \sum_{e \in E_0} u_e x_e - \sum_{e \in E_T} (\ell_e - \ell_e x_e) \geq r.
\end{align*}
%
Here, the last inequality holds because by our setting of $\alpha$, we have 
$$\frac{\alpha - 1 - \beta}{\beta} = 4\Gamma' + \kappa + 1 + \epsilon,$$ 
and thus Lemma~\ref{lemma:differenceapprox2} gives that $$\sum_{e \in E_T} (\ell_e - \ell_e x_e) \leq \frac{\alpha - 1 - \beta}{\beta}\sum_{e \in \sol} (\ell_e - \ell_e x_e) \leq \frac{\alpha - 1 - \beta}{\beta}\sum_{e \in \sol} \ell_e.$$
\end{proof}

By using Lemma~\ref{lemma:steinersep2} with the ellipsoid method and the fact that the LP optimum is at most $\mr$, we get an $(\alpha, \beta)$-robust fractional solution. Then, Theorem~\ref{thm:generalrounding} and known approximation/integrality gap results give us the following theorem, which with appropriate choice of constants gives Theorem~\ref{thm:steiner-main}:

\begin{theorem}\label{thm:steiner}
Fix any $\Gamma' > 1, \kappa > 0, 0 < \epsilon < 4/35$ and let $\Gamma$ be the constant given in Lemma~\ref{lemma:differenceapprox2}. Let $\alpha = (4\Gamma'+\kappa+2+\epsilon)4\Gamma+1$ and $\beta = 4\Gamma$. Then there exists a polynomial-time $(2\alpha \ln 4 + \epsilon, 2 \beta \ln 4 + \ln 4+\epsilon)$-robust algorithm for the Steiner tree problem.
\end{theorem}

\begin{proof}[Proof of Theorem~\ref{thm:steiner}]
By using the ellipsoid method with Lemma~\ref{lemma:steinersep2} we can compute a feasible $(\alpha, \beta)$-robust fractional solution to the Steiner tree LP (as the robust Steiner tree LP has optimum at most $\mr$). Then, the theorem follows from Theorem~\ref{thm:generalrounding}, and the fact that the polytope in Fig.~\ref{fig:rrst} has integrality gap $\delta = 2$ and there is a $\gamma = (\ln 4 + \epsilon)$-approximation for the Steiner tree problem due to \cite{ByrkaGRS10} (The error parameters can be rescaled appropriately to get the approximation guarantee in the theorem statement).
\end{proof}

Optimizing for $\alpha$ in Theorem~\ref{thm:steiner} subject to the constraints in \eqref{eq:maincond}, we get that for a fixed (small) $\epsilon$, $\alpha$ is minimized by setting $\Gamma \approx 9.284 + f_1(\epsilon), \Gamma' \approx 5.621 + f_2(\epsilon), \kappa \approx 2.241 + f_3(\epsilon)$ (for monotonically increasing $f_1, f_2, f_3$ which approach $0$ as $\epsilon$ approaches 0). Plugging in these values gives Theorem~\ref{thm:steiner-main}.

\subsection{Proof of Lemma~\ref{lemma:differenceapprox2}}\label{section:differenceapprox2}

We will again use local search to find moves that are improving with respect to $c$. However, now our goal is to show that we can do this without blowing up the cost with respect to $c'$. We can start to show this via the following lemma, which generalizes the arguments in Section~\ref{section:localsearch}. Informally, it says that as long as a significant fraction $(1/\theta)$ of the swaps (rather than all the swaps) that the local search algorithm can make between its solution $A$ and an adversarial solution $F$ do not improve its objective by some factor $\lambda$ (rather than by any amount at all), $A$'s cost can still be bounded by $4 \lambda \theta$ times $F$'s cost.

From Lemma~\ref{lemma:approxlocal} to Lemma~\ref{lemma:apxsp} we will refer to the cost functions on edges by $w, w'$ instead of $c, c'$. This is because these lemmas are agnostic to the cost functions they are applied to and will be applied with both $w = c, w' = c'$ and $w = c', w' = c$ in our algorithm. We also define $\mathcal{A}, \mathcal{F}, \alpha, N(\cdot), N^{-1}(\cdot)$ as in the proof of Theorem~\ref{thm:localsearch} for these lemmas.

\begin{lemma}\label{lemma:approxlocal}
Let $A$ and $F$ be solutions to an instance of Steiner tree with edge costs $w$ such that if all edges in $A \cap F$ have their costs set to 0, then for $\lambda \geq 1, \theta \geq 1$, we have

$$\sum_{a \in \mathcal{A}, f \in N(a): w(a) \leq \lambda w(f)} \alpha(a, f)w(a) \geq \frac{1}{\theta} \sum_{a \in \mathcal{A}, f \in N(a)} \alpha(a, f)w(a).$$

Then $w(A \setminus F) \leq 4\lambda\theta w(F \setminus A)$.
\end{lemma}

\begin{proof}

This follows by generalizing the argument in the proof of Theorem~\ref{thm:localsearch}.  After setting costs of edges in $A \cap F$ to 0, note that $w(A) = w(A \setminus F)$ and $w(F) = w(F \setminus A)$. Then:

\begin{align*}
w(A \setminus F) &= \sum_{a \in \mathcal{A}} w(a)\\
&\leq \sum_{a \in \mathcal{A}} \sum_{f \in N(a)} \alpha(a, f) w(a) \\
&= \sum_{f \in F} \sum_{a \in N^{-1}(f)} \alpha(a, f) w(a) \\
&\leq \theta \sum_{f \in F} \sum_{a \in N^{-1}(f): w(a) \leq \lambda w(f)} \alpha(a, f) w(a) \\
&\leq \lambda\theta \sum_{f \in F} \sum_{a \in N^{-1}(f): w(a) \leq \lambda w(f)} \alpha(a, f) w(f)\\ &\leq \lambda\theta \sum_{f \in F} \sum_{a \in N^{-1}(f)} \alpha(a, f) w(f) \leq 4\lambda\theta w(F \setminus A).
\end{align*}
\end{proof}

\begin{corollary}\label{cor:approxlocal}
Let $A$, $F$ be solutions to an instance of Steiner tree with edge costs $w$ such that for parameters $\lambda \geq 1, \theta \geq 1$, $w(A \setminus F) > 4\lambda\theta w(F \setminus A)$. Then after setting the cost of all edges in $A \cap F$ to 0, $$\sum_{a \in \mathcal{A}, f \in N(a): w(a) > \lambda w(f)} \alpha(a, f)w(a) > \frac{\theta - 1}{\theta} \sum_{a \in \mathcal{A}, f \in N(a)} \alpha(a, f)w(a).$$
\end{corollary}

The corollary effectively tells us that if $w(A \setminus F)$ is sufficiently larger than $w(F \setminus A)$, then there are many local swaps between $S$ in $A$ and $f$ in $F$ that decrease $w(A)$ by a large fraction of $w(a)$. The next lemma then shows that one of these swaps also does not increase $w'(A)$ by a large factor (even if instead of swapping in $f$, we swap in an approximation of $f$), and reduces $w(A)$ by a non-negligible amount.

\begin{lemma}\label{lemma:greedyswap}
Let $A$ and $F$ be solutions to an instance of Steiner tree with two sets of edge costs, $w$ and $w'$, such that for parameter $\Gamma > 1$, $w(A \setminus F) > 4 \Gamma \cdot w(F \setminus A)$. Fix any $0 < \epsilon < \sqrt{\Gamma} - 1$. Then there exists a swap between $a \in \mathcal{A}$ and a path $f$ between two vertices in $A$ such that $\frac{(1+\epsilon)w'(f) - w'(a)}{w(a) - (1+\epsilon)w(f)} \leq \frac{4(1+\epsilon)\Gamma}{(\sqrt{\Gamma}-1)(\sqrt{\Gamma}-1-\epsilon)} \cdot \frac{w'(F \setminus A)}{w(A \setminus F)}$ and $w(a) - (1+\epsilon)w(f) \geq \frac{1}{n^2} w(A \setminus F)$.
\end{lemma}
\begin{proof}
We use an averaging argument to prove the lemma. Consider the quantity 
$$R = \frac{\sum_{a \in \mathcal{A}, f \in N(a): w(a) > \sqrt{\Gamma} w(f), w(a) -(1+\epsilon)w(f) \geq \frac{1}{n^2} w(A \setminus F)} \alpha(a, f) [(1+\epsilon)w'(f) - w'(a)]}{\sum_{a \in \mathcal{A}, f \in N(a): w(a) > \sqrt{\Gamma} w(f), w(a) - (1+\epsilon)w(f) \geq \frac{1}{n^2} w(A \setminus F)} \alpha(a, f) [w(a) - (1+\epsilon)w(f)]},$$ 
which is the ratio of the weighted average of increase in $c'$ to the weighted average of decrease in $c$ over all swaps where $w(a) > \sqrt{\Gamma} w(f)$ and $w(a) -(1+\epsilon)w(f) \geq \frac{1}{n^2} w(A \setminus F)$.

For any edge $e$ in $A \cap F$, it is also a subpath $f \in \mathcal{A} \cap \mathcal{F}$ for which the only $a \in \mathcal{A}$ such that $A \cup f \setminus a$ is feasible is $a = f$. So for all such $e$ we can assume that $\alpha$ is defined such that $\alpha(e, e) = 1$, $\alpha(e, f) = 0$ for $f \neq e$, $\alpha(a, e) = 0$ for $a \neq e$. Clearly $w(a) > \sqrt{\Gamma}w(a)$ does not hold, so no swap with a positive $\alpha$ value in either sum involves edges in $A \cap F$. So we can now set the cost with respect to both $c, c'$ of edges in $A \cap F$ to 0, and doing so does not affect the quantity $R$.

Then, the numerator can be upper bounded by $4(1+\epsilon)w'(F \setminus A)$. For the denominator, we first observe that
\begin{multline}\label{eq:greedyswap}
\sum_{a \in \mathcal{A}, f \in N(a): w(a) > \sqrt{\Gamma} w(f), w(a) -(1+\epsilon)w(f) \geq \frac{1}{n^2} w(A \setminus F)} \alpha(a, f) [w(a) - (1+\epsilon)w(f)] \geq \\
\sum_{a \in \mathcal{A}, f \in N(a): w(a) > \sqrt{\Gamma} w(f)} \alpha(a, f) [w(a) - (1+\epsilon)w(f)] - \\
\sum_{a \in \mathcal{A}, f \in N(a): w(a) -(1+\epsilon)w(f) < \frac{1}{n^2} w(A \setminus F)} \alpha(a, f) [w(a) - (1+\epsilon)w(f)].
\end{multline}

The second term on the right-hand side of \eqref{eq:greedyswap} is upper bounded by:

$$\sum_{a \in \mathcal{A}, f \in N(a): w(a) -(1+\epsilon)w(f) < \frac{1}{n^2} w(A \setminus F)} \alpha(a, f) [w(a) - (1+\epsilon)w(f)] \leq $$ 
$$\frac{1}{n^2} \sum_{a \in \mathcal{A}, f \in N(a): w(a) -(1+\epsilon)w(f) < \frac{1}{n^2} w(A \setminus F)} \alpha(a, f) w(A \setminus F)\leq \frac{1}{n}w(A \setminus F).$$

The inequality follows because there are at most $n$ different $a \in \mathcal{A}$, and for each one $\sum_{f \in F} \alpha(a, f) = 1$. We next use Corollary~\ref{cor:approxlocal} (setting both parameters to $\sqrt{\Gamma}$) to get the following lower bound on the first term in~\eqref{eq:greedyswap}:

\begin{align*}
&\sum_{a \in \mathcal{A}, f \in N(a): w(a) > \sqrt{\Gamma} w(f)} \alpha(a, f) [w(a) - (1+\epsilon)w(f)]\\
&\geq \sum_{a \in \mathcal{A}, f \in N(a): w(a) > \sqrt{\Gamma} w(f)} \alpha(a, f) [w(a) - \frac{1+\epsilon}{\sqrt{\Gamma}}w(a)]\\ 
&=\frac{\sqrt{\Gamma}-1-\epsilon}{\sqrt{\Gamma}} \sum_{a \in \mathcal{A}, f \in N(a): w(a) > \sqrt{\Gamma} w(f)} \alpha(a, f) w(a)\\
&\geq 
\frac{(\sqrt{\Gamma}-1)(\sqrt{\Gamma}-1-\epsilon)}{\Gamma} \sum_{a \in \mathcal{A}, f \in N(a)} \alpha(a, f) w(a) \\
&=\frac{(\sqrt{\Gamma}-1)(\sqrt{\Gamma}-1-\epsilon)}{\Gamma} w(A \setminus F).
\end{align*}

Which lower bounds the denominator of $R$ by $\left(\frac{(\sqrt{\Gamma}-1)(\sqrt{\Gamma}-1-\epsilon)}{\Gamma} - 1/n\right) \cdot w(A \setminus F)$. By proper choice of $\epsilon$, for sufficiently large $n$ we can ignore the $1/n$ term. Then, combining the bounds implies that $R$ is at most $\frac{4(1+\epsilon')\Gamma}{(\sqrt{\Gamma}-1)(\sqrt{\Gamma}-1-\epsilon')} \cdot \frac{w'(F \setminus A)}{w(A \setminus F)}$. In turn, one of the swaps being summed over in $R$ satisfies the lemma statement.
\end{proof}

We now almost have the tools to prove Lemma~\ref{lemma:differenceapprox2}. However, the local search process is now concerned with two edge costs, so just considering adding the shortest path with respect to $c$ between each pair of vertices and deleting a subset of vertices in the induced cycle will not suffice. We instead use the following lemma:

\begin{lemma}\label{lemma:apxsp}
Given a graph $G=(V, E)$ with edge costs $w$ and $w'$, two vertices $s$ and $t$, and input parameter $W'$, let $p$ be the shortest path from $s$ to $t$ with respect to $w$ whose cost with respect to $c'$ is at most $W'$. For all $\epsilon > 0$, there is a polynomial time algorithm that finds a path from $s$ to $t$ whose cost with respect to $w$ is at most $w(p)$ and whose cost with respect to $w'$ is at most $(1+\epsilon)W'$.
\end{lemma}
\begin{proof}
If all edge lengths with respect to $w'$ are multiples of $\Delta$, an optimal solution can be found in time $\textnormal{poly}(|V|, |E|, W'/\Delta)$ via dynamic programming: Let $\ell(v, i)$ be the length of the shortest path from $s$ to $v$ with respect to $w$ whose cost with respect to $w'$ is at most $i \cdot \Delta$. Using $\ell(s, i) = 0$ for all $i$ and the recurrence $\ell(v, i) \leq \ell(u, i-(w'_{uv}/\Delta)) + w_{uv}$ for edge $(u, v)$, we can compute $\ell(v, i)$ for all $v, i$ and use backtracking from $\ell(t, W')$ to retrieve $p$ in $\textnormal{poly}(|V|, |E|, W'/\Delta)$ time.

To get the runtime down to polynomial, we use a standard rounding trick, rounding each $w_e'$ down to the nearest multiple of $\epsilon W' / |V| $. After rounding, the runtime of the dynamic programming algorithm is $\textnormal{poly}(|V|, |E|, \frac{W'}{\epsilon W'/|V|}) = \textnormal{poly}(|V|, |E|, \frac{1}{\epsilon})$. Any path has at most $|V|$ edges, and so its cost decreases by at most $\epsilon W'$ in this rounding process, i.e., all paths considered by the algorithm have cost with respect to $w'$ of at most $(1+\epsilon)W'$. Lastly, since $p$'s cost with respect to $w'$ only decreases, $w(p)$ still upper bounds the cost of the shortest path considered by the algorithm with respect to $w$.
\end{proof}

The idea is to run a local search with respect to $c$ starting with a good approximation with respect to $c'$. Our algorithm alternates between a ``forward'' and ``backward'' phase. In the forward phase, we use Lemma~\ref{lemma:apxsp} to decide which paths can be added to the solution in local search moves. The local search takes any swap that causes both $c(\alg)$ and $c'(\alg)$ to decrease if any exists. Otherwise, it picks the swap between $S \in \alg$ and $f$ that among all swaps where $c(f) < c(a)$ and  $c'(f) \leq c'(\sol)$ minimizes the ratio $\frac{c'(f) - c'(a)}{c(a) - c(f)}$ (we assume we know the value of $c'(\sol)$, as can guess many values, and our algorithm will work for the right value for $c'(\sol)$). 

If the algorithm only made swaps of this form, however, $c'(\alg)$ might become a very poor approximation of $c'(\sol)$. To control for this, when $c'(\alg)$ exceeds $(4\Gamma'+\kappa) \cdot c'(\sol)$ for some constant $\Gamma' > 1$, we begin a ``backward phase'': We take the opposite approach, greedily choosing either swaps that improve both $c$ and $c'$ or that improve $c'$ and minimize the ratio $\frac{c(f) - c(a)}{c'(a) - c'(f)}$, until $c'(\alg)$ has been reduced by at least $\kappa \cdot c'(\sol)$. At this point, we begin a new forward phase. 

The intuition for the analysis is 
as follows: If, throughout a forward phase, $c(\alg \setminus \sol) \geq 4\Gamma \cdot c(\sol \setminus \alg)$, Lemma~\ref{lemma:greedyswap} tells us that there is  swap where the increase in $c'(\alg)$ will be very small relative to the decrease in $c(\alg)$.  (Note that our goal is to reduce the cost of $c(\alg \setminus \sol)$ to something below $4\Gamma \cdot c(\sol \setminus \alg)$.)
Throughout the subsequent backward phase, we have $c'(\alg) > 4\Gamma' \cdot c'(\sol)$, which implies $c'(\alg \setminus \sol) > 4\Gamma' \cdot c'(\sol \setminus \alg)$. So Lemma~\ref{lemma:greedyswap} also implies that the total increase in $c(\alg)$ will be very small relative to the decrease in $c'(\alg)$. Since the absolute change in $c'(\alg)$ is similar between the two phases, one forward and one backward phase should decrease $c(\alg)$ overall. 

\begin{figure}[H]
\fbox{\begin{minipage}{.98\textwidth}
\textsc{DoubleApprox}($G=(V, E)$, $\chi$, $\chi'$, $\Gamma$, $\Gamma'$, $\kappa$)

\begin{algorithm}[H]
\KwData{A graph $G=(V, E)$ with terminal set $T$ and cost functions $c, c'$ for which all $c_e'$ are a multiple of $\frac{\epsilon}{n}\chi'$, $\chi$ such that $\chi \in [c(\sol), (1+\epsilon)\cdot c(\sol)]$, $\chi'$ such that $\chi' \in [c'(\sol), (1+\epsilon)\cdot c'(\sol)]$, constants $\Gamma, \Gamma', \kappa$.}
$i \leftarrow 0$\;
$\alg^{(0)} \leftarrow $ $\alpha$-approximation of optimal Steiner tree with respect to $c'$ for $\alpha < 4\Gamma'$ \label{line:approx}\;
\While{$i = 0$ or $c(\alg^{(i)}) < c(\alg^{(i-2)})$}{\label{line:mainwhile} \tcc{Iterate over all guesses $\rho$ for $c(\alg^{(i)} \setminus \sol)$}
\For{$\rho \in \{\min_e c_e, (1+\epsilon)\min_e c_e, \ldots (1+\epsilon)^{\log_{1+\epsilon} n \frac{\max_e c_e}{\min_e c_e}} \min_e c_e\}$}{
$\alg^{(i+1)}_\rho \leftarrow \alg^{(i)}$\;
\tcc{Forward phase}
\While{$c'(\alg^{(i+1)}_\rho) \leq (4\Gamma'+\kappa)\chi'$ and $c(\alg^{(i+1)}_\rho) > c(\alg^{(i)}) - \rho/2$}{\label{line:firstwhile} 
$(\alg^{(i+1)}_\rho, stop)\leftarrow \textsc{GreedySwap}(\alg^{(i+1)}_\rho, c, c', \chi', \frac{1}{10n^2}\rho)$\;
\If{$stop = 1$}{\label{line:stopcheckfirst}\
\textbf{break} while loop starting on line~\ref{line:firstwhile}\;
}\label{line:stopchecklast}
}
$\alg^{(i+2)}_\rho \leftarrow \alg^{(i+1)}_\rho$\;
\tcc{Backward phase}
\While{$c'(\alg^{(i+2)}_\rho) \geq 4\Gamma'\chi'$}{\label{line:secondwhile} 
\label{line:secondswap} $(\alg^{(i+2)}_\rho, \sim) \leftarrow \textsc{GreedySwap}(\alg^{(i+2)}_\rho, c', c, \chi, \frac{\epsilon}{n}\chi')$\;
}
}
$\alg^{(i+2)} \leftarrow \argmin_{\alg^{(i+2)}_\rho} c(\alg^{(i+2)}_\rho)$\;
$i \leftarrow i + 2$\;
}
\textbf{return} all values of $\alg^{(i)}_\rho$ stored for any value of $i, \rho$\;
\end{algorithm}
\end{minipage}}
\caption{Algorithm \textsc{DoubleApprox}, which finds $\alg$ such that $c(\alg \setminus \sol) \leq O(1) \cdot c(\sol \setminus \alg)$ and $c'(\alg) \leq O(1) \cdot c'(\sol)$}
\label{fig:doubleapprox}
\end{figure}

\begin{figure}[H]
\fbox{\begin{minipage}{.98\textwidth}
$\textsc{GreedySwap}(\alg, w, w', \chi', \rho)$

\begin{algorithm}[H]
\KwData{Solution $\alg$, cost functions $w, w'$ on the edges, $\chi' \in [w'(\sol), (1+\epsilon)w'(\sol)]$, minimum improvement per swap $\rho$}
$swaps \leftarrow \emptyset$\;
\For{$W' \in \{1, 1+\epsilon, (1+\epsilon)^2, \ldots, (1+\epsilon)^{\lceil \log_{1 + \epsilon} \chi' \rceil}\}$}{
\For{$s, t \in \alg$}{
Find a $(1+\epsilon)$-approximation $f$ of the shortest path $\hat{f}$ from $s$ to $t$ with respect to $w$ such that $w'(\hat{f}) \leq W'$, $\hat{f} \cap \alg = \{s, t\}$ (via Lemma~\ref{lemma:apxsp})\;
\label{line:increasebound}
\For{all maximal $a \subseteq \alg$ such that $\alg \cup a \setminus f$ is feasible, $w(a) - w(f) \geq \rho$}{
$swaps \leftarrow swaps \cup \{(a, f)\}$\;
}
}
}
\If{$swaps = \emptyset$}{
\textbf{return} $(\alg, 1)$\;
}
$(a^*, f^*) \leftarrow \argmin_{(a, f) \in swaps} \frac{w'(f)-w'(a)}{w(a)-w(f)}$\;
\textbf{return} $(\alg \cup a^* \setminus f^*, 0)$\;
\end{algorithm}
\end{minipage}}
\caption{Algorithm \textsc{GreedySwap}, which finds a swap with the properties described in Lemma~\ref{lemma:greedyswap}}
\label{fig:greedyswap}
\end{figure}

The formal description of the backward and forward phase is given as algorithm \textsc{DoubleApprox} in Figure~\ref{fig:doubleapprox}. 

For the lemmas/corollaries stated, we implicitly assume that we know values of $\chi$ and $\chi'$ satisfying the conditions of \textsc{DoubleApprox}. When we conclude by proving Lemma~\ref{lemma:differenceapprox2}, we will simply call \textsc{DoubleApprox} for every reasonable value of $\chi, \chi'$ that is a power of $1+\epsilon$, and one of these runs will have $\chi, \chi'$ satisfying the conditions. Furthermore, there are multiple error parameters in our algorithm and its analysis. For simplicity of presentation, we will use the same value $\epsilon$ for all error parameters. We now begin the analysis. We first make some observations. The first lets us relate the decrease in cost of a solution $\alg$ to the decrease in the cost of $\alg \setminus \sol$.

\begin{observation}\label{obs:swap}
Let $\alg, \alg', \sol$ be any Steiner tree solutions to a given instance. Then 

$$c(\alg) - c(\alg') = [c(\alg \setminus \sol) - c(\sol \setminus \alg)] - [c(\alg' \setminus \sol) - c(\sol \setminus \alg')].$$
\end{observation}
\begin{proof}
By symmetry, the contribution of edges in $\alg \cap \alg'$ and edges in neither $\alg$ nor $\alg'$ to both the left and right hand side of the equality is zero, so it suffices to show that all edges in $\alg \oplus \alg'$ contribute equally to the left and right hand side.

Consider any $e \in \alg \setminus \alg'$. Its contribution to $c(\alg) - c(\alg')$ is $c(e)$. If $e \in \alg \setminus \sol$, then $e$ contributes $c(e)$ to $c(\alg \setminus \sol) - c(\sol \setminus \alg)$ and 0 to $-[c(\alg' \setminus \sol) - c(\sol \setminus \alg')]$. If $e \in \alg \cap \sol$, then $e$ contributes $0$ to $c(\alg \setminus \sol) - c(\sol \setminus \alg)$ and $c(e)$ to $- [c(\alg' \setminus \sol) - c(\sol \setminus \alg')]$. So the total contribution of $e$ to $[c(\alg \setminus \sol) - c(\sol \setminus \alg)] - [c(\alg' \setminus \sol) - c(\sol \setminus \alg')]$ is $c(e)$.

Similarly, consider $e \in \alg' \setminus \alg$. Its contribution to $c(\alg) - c(\alg')$ is $-c(e)$. If $e \in \sol \setminus \alg$, then $e$ contributes $-c(e)$ to $c(\alg \setminus \sol) - c(\sol \setminus \alg)$ and $0$ to $[c(\alg' \setminus \sol) - c(\sol \setminus \alg')]$. If $e \notin \sol$, then $e$ contributes $0$ to $c(\alg \setminus \sol) - c(\sol \setminus \alg)$ and $-c(e)$ to $-[c(\alg' \setminus \sol) - c(\sol \setminus \alg')]$. So the total contribution of $e$ to $[c(\alg \setminus \sol) - c(\sol \setminus \alg)] - [c(\alg' \setminus \sol) - c(\sol \setminus \alg')]$ is $-c(e)$.
\end{proof}

Observation~\ref{obs:swap} is useful because Lemma~\ref{fig:greedyswap} relates the ratio of change in $c, c'$ to the $c(\alg \setminus \sol)$, but it is difficult to track how $c(\alg \setminus \sol)$ changes as we make swaps that improve $c(\alg)$. For example, $c(\alg \setminus \sol)$ does not necessarily decrease with swaps that cause $c(\alg)$ to decrease (e.g. consider a swap that adds a light edge not in $\sol$ and removes a heavy edge in $\sol$). Whenever $c(\alg \setminus \sol) \gg c(\sol \setminus \alg)$ (if this doesn't hold, we have a good approximation and are done), $c(\alg \setminus \sol)$ and $c(\alg \setminus \sol) - c(\sol \setminus \alg)$ are off by a multiplicative factor that is very close to 1, and thus we can relate the ratio of changes in Lemma~\ref{fig:greedyswap} to $c(\alg \setminus \sol) - c(\sol \setminus \alg)$ instead at a small loss in the constant, and by Observation~\ref{obs:swap} changes in this quantity are much easier to track over the course of the algorithm, simplifying our analysis greatly.

The next observation lets us assume that any backward phase requires only polynomially many calls to \textsc{GreedySwap}.

\begin{observation}\label{obs:rounding}
Let $\chi'$ be any value such that $\chi' \in [c'(\sol), (1+\epsilon)c'(\sol)]$, and suppose we round all $c_e'$ up to the nearest multiple of $\frac{\epsilon}{n}\chi'$ for some $0 < \epsilon < 1$. Then any $\gamma$-approximation of $\sol$ with respect to $c'$ using the rounded $c_e'$ values is an $\gamma(1+2\epsilon)$-approximation of $\sol$ with respect to $c'$ using the original edge costs.
\end{observation}
\begin{proof}
This follows because the rounding can only increase the cost of any solution, and the cost increases by at most $\epsilon \chi' \leq \epsilon(1 + \epsilon)c'(\sol) \leq 2 \epsilon c'(\sol)$.
\end{proof}

Via this observation, we will assume all $c'_e$ are already rounded.

\begin{lemma}[Forward Phase Analysis]\label{lemma:forwardphase}
For any even $i$ in algorithm \textsc{DoubleApprox}, let $\rho$ be the power of $(1+\epsilon)$ times $\min_e c_e$ such that $\rho \in [c(\alg^{(i)} \setminus \sol),(1+\epsilon)c(\alg^{(i)} \setminus \sol) ]$. Suppose all values of $\alg^{(i+1)}_\rho$ and the final value of $\alg^{(i)}$ in \textsc{DoubleApprox} satisfy $c(\alg^{(i+1)}_\rho \setminus \sol) > 4\Gamma \cdot c(\sol \setminus \alg^{(i+1)}_{\rho})$ and $c(\alg^{(i)} \setminus \sol) > 4\Gamma \cdot c(\sol \setminus \alg^{(i)})$. Then for $0 < \epsilon < 2/3 -5/12\Gamma$, the final values of $\alg^{(i)}, \alg^{(i+1)}_\rho$ satisfy 

$$c(\alg^{(i)}) - c(\alg^{(i+1)}_\rho) \geq \min\left\{\frac{4\Gamma - 1}{8\Gamma} ,\frac{(4\Gamma - 1)(\sqrt{\Gamma}-1)(\sqrt{\Gamma}-1-\epsilon)\kappa}{16(1+\epsilon)\Gamma^2}\right\} \cdot c(\alg^{(i+1)}_\rho \setminus \sol).$$
\end{lemma}

\begin{proof}

Let $\alg^{(i+1)}_{\rho, j}$ denote the value of $\alg^{(i+1)}_\rho$ after $j$ calls to \textsc{GreedySwap} on $\alg^{(i+1)}_\rho$, and let $J$ be the total number of calls of \textsc{GreedySwap} on $\alg^{(i+1)}_\rho$. Then $\alg^{(i+1)}_{\rho, 0}$ is the final value of $\alg^{(i)}$, and the final value of $\alg^{(i+1)}_\rho$ is $\alg^{(i+1)}_{\rho, J}$. 
Any time \textsc{GreedySwap} is invoked on $\alg^{(i+1)}_\rho$, by line~\ref{line:firstwhile} of \textsc{DoubleApprox} and the assumption  $\rho \leq (1+\epsilon)c(\alg^{(i)} \setminus \sol)$ in the lemma statement, we have:

$$c(\alg^{(i+1)}_\rho) > c(\alg^{(i)}) - \rho/2 \geq  c(\alg^{(i)}) - \frac{1+\epsilon}{2}c(\alg^{(i)} \setminus \sol).$$

Then, by Observation~\ref{obs:swap} and the assumption  $c(\alg^{(i)} \setminus \sol) > 4\Gamma \cdot c(\sol \setminus \alg^{(i)})$ in the lemma statement, we have:

\begin{align*}
c(\alg^{(i+1)}_\rho \setminus \sol) &\geq c(\alg^{(i+1)}_\rho \setminus \sol) - c(\sol \setminus \alg^{(i+1)}_{\rho}) \\
&=c(\alg^{(i)} \setminus \sol) - c(\sol \setminus \alg^{(i)}) + c(\alg^{(i+1)}_\rho) - c(\alg^{(i)}) \\
&\geq c(\alg^{(i)} \setminus \sol) - c(\sol \setminus \alg^{(i)}) - \frac{1+\epsilon}{2}c(\alg^{(i)} \setminus \sol) \\
& \geq \left(\frac{1-\epsilon}{2} - \frac{1}{4\Gamma}\right)c(\alg^{(i)} \setminus \sol),
\end{align*}

For $\epsilon < 2/3 - 5/12\Gamma$,  $c(\alg^{(i+1)}_{\rho} \setminus \sol)/n^2 \geq \frac{1}{10n^2}\rho$. So by Lemma~\ref{lemma:greedyswap} \textsc{GreedySwap} never outputs a tuple where $stops = 1$, and thus we can ignore lines~\ref{line:stopcheckfirst}-\ref{line:stopchecklast} of \textsc{DoubleApprox} under the conditions in the lemma statement.

Suppose $\alg^{(i+1)}_{\rho, J}$ satisfies $c(\alg^{(i+1)}_{\rho, J}) \leq c(\alg^{(i)}) - \rho/2$, a condition that causes the while loop at line~\ref{line:firstwhile} of \textsc{DoubleApprox} to exit and the forward phase to end. Then

\begin{align*}
c(\alg^{(i)}) - c(\alg^{(i+1)}_{\rho, J}) &\geq \rho/2 \geq \frac{1}{2} c(\alg^{(i)} \setminus \sol)\\
&\geq \frac{1}{2} [c(\alg^{(i)} \setminus \sol) - c(\sol \setminus \alg^{(i)})]\\
&= \frac{1}{2} [c(\alg^{(i+1)}_{\rho,0} \setminus \sol) - c(\sol \setminus \alg^{(i+1)}_{\rho,0})]\\
&\geq \frac{1}{2}[c(\alg^{(i+1)}_{\rho,J} \setminus \sol) - c(\sol \setminus \alg^{(i+1)}_{\rho,J}) ]\\
&\geq \frac{4\Gamma - 1}{8\Gamma} c(\alg^{(i+1)}_{\rho, J}\setminus \sol).
\end{align*}

The second-to-last inequality is using Observation~\ref{obs:swap}, which implies $c(\alg^{(i+1)}_{\rho, j} \setminus \sol) - c(\sol \setminus \alg^{(i+1)}_{\rho, j})$ is decreasing with swaps, and the last inequality holds by the assumption $c(\alg^{(i+1)}_\rho \setminus \sol) > 4\Gamma \cdot c(\sol \setminus \alg^{(i+1)}_{\rho})$ in the lemma statement. Thus if $c(\alg^{(i+1)}_{\rho, J}) \leq c(\alg^{(i)}) - \rho/2$, the lemma holds.

Now assume instead that $c(\alg^{(i+1)}_{\rho, J}) > c(\alg^{(i)}) - \rho/2$ when the forward phase ends. We want a lower bound on
$$c(\alg^{(i+1)}_{\rho, 0}) - c(\alg^{(i+1)}_{\rho, J})= \sum_{j = 0}^{J-1}[c(\alg^{(i+1)}_{\rho, j}) - c(\alg^{(i+1)}_{\rho, j+1})].$$

We bound each $c(\alg^{(i+1)}_{\rho, j}) - c(\alg^{(i+1)}_{\rho, j+1})$ term using Lemma~\ref{lemma:greedyswap} and Lemma~\ref{lemma:apxsp}. By Lemma~\ref{lemma:greedyswap} and the assumption in the lemma statement that  $c(\alg^{(i+1)}_\rho \setminus \sol) > 4\Gamma \cdot c(\sol \setminus \alg^{(i+1)}_{\rho})$, we know there exists a swap between $a \in \alg_{\rho, j}^{(i+1)}$ and $f \in \sol$ such that

$$\frac{(1+\epsilon)c'(f) - c'(a)}{c(a) - (1+\epsilon)c(f)} \leq \frac{4(1+\epsilon)\Gamma}{(\sqrt{\Gamma}-1)(\sqrt{\Gamma}-1-\epsilon)} \cdot \frac{c'(\sol \setminus \alg^{(i+1)}_{\rho, j})}{c(\alg^{(i+1)}_{\rho, j} \setminus \sol)}.$$ 

By Lemma~\ref{lemma:apxsp}, we know that when $G'$ is set to a value in $[c'(f), (1+\epsilon) \cdot c'(f)]$ in line 2 of in \textsc{GreedySwap}, the algorithm finds a path $f'$ between the endpoints of $f$ such that $c(f') \leq (1+\epsilon) c(f)$ and $c'(f') \leq (1+\epsilon) c'(f)$. Thus $(a, f') \in swaps$ and the swap $(a^*, f^*)$ chosen by the $(j+1)$th call to \textsc{GreedySwap} satisfy:

$$\frac{c'(f^*) - c'(a^*)}{c(a^*) - c(f^*)}
\leq \frac{c'(f') - c'(a)}{c(a) - c(f)} 
\leq \frac{(1+\epsilon)c'(f) - c'(a)}{c(a) - (1+\epsilon)c(f)} 
\leq$$
$$\frac{4(1+\epsilon)\Gamma}{(\sqrt{\Gamma}-1)(\sqrt{\Gamma}-1-\epsilon)} \cdot \frac{c'(\sol \setminus \alg^{(i+1)}_{\rho, j})}{c(\alg^{(i+1)}_{\rho, j} \setminus \sol)}.$$

Rearranging terms and observing that $c'(\sol) \geq c'(\sol \setminus \alg^{(i+1)}_{\rho,j})$ gives:

\begin{align*}
c(\alg^{(i+1)}_{\rho, j}) - c(\alg^{(i+1)}_{\rho, j+1}) &= c(a^*) - c(f^*)\\
&\geq \frac{(\sqrt{\Gamma}-1)(\sqrt{\Gamma}-1-\epsilon)}{4(1+\epsilon)\Gamma} \cdot c(\alg^{(i+1)}_{\rho,j} \setminus \sol) \frac{c'(f^*)-c'(a^*)}{c'(\sol)}\\
&= \frac{(\sqrt{\Gamma}-1)(\sqrt{\Gamma}-1-\epsilon)}{4(1+\epsilon)\Gamma} \cdot c(\alg^{(i+1)}_{\rho,j} \setminus \sol) \frac{c'(\alg^{(i+1)}_{\rho,j+1})-c'(\alg^{(i+1)}_{\rho,j})}{c'(\sol)}.
\end{align*}

This in turn gives:

\begin{align*}
c(\alg^{(i+1)}_{\rho, 0}) - c(\alg^{(i+1)}_{\rho,J})&= \sum_{j = 0}^{J-1}[c(\alg^{(i+1)}_{\rho, j}) - c(\alg^{(i+1)}_{\rho, j+1})] \\
&\geq \sum_{j=0}^{J-1} \frac{(\sqrt{\Gamma}-1)(\sqrt{\Gamma}-1-\epsilon)}{4(1+\epsilon)\Gamma} \cdot c(\alg^{(i+1)}_{\rho,j} \setminus \sol) \frac{c'(\alg^{(i+1)}_{\rho, j+1})-c'(\alg^{(i+1)}_{\rho, j})}{c'(\sol)} \\
& \geq \frac{(\sqrt{\Gamma}-1)(\sqrt{\Gamma}-1-\epsilon)}{4(1+\epsilon)\Gamma} \sum_{j=0}^{J-1} [c(\alg^{(i+1)}_{\rho, j} \setminus \sol) - c(\sol \setminus \alg^{(i+1)}_{\rho, j})]\\
&\qquad\qquad \cdot \frac{c'(\alg^{(i+1)}_{\rho, j+1})-c'(\alg^{(i+1)}_{\rho, j})}{c'(\sol)} \\
&\geq \frac{(\sqrt{\Gamma}-1)(\sqrt{\Gamma}-1-\epsilon)}{4(1+\epsilon)\Gamma} \sum_{j=0}^{J-1} [c(\alg^{(i+1)}_{\rho, J} \setminus \sol) - c(\sol \setminus \alg^{(i+1)}_{\rho, J})]\\
&\qquad\qquad\cdot \frac{c'(\alg^{(i+1)}_{\rho, j+1})-c'(\alg^{(i+1)}_{\rho, j})}{c'(\sol)}\\
&\geq \frac{(4\Gamma - 1)(\sqrt{\Gamma}-1)(\sqrt{\Gamma}-1-\epsilon)}{16(1+\epsilon)\Gamma^2} c(\alg^{(i+1)}_{\rho, J} \setminus \sol) \\
&\qquad\qquad \cdot \sum_{j=0}^{J-1}  \frac{c'(\alg^{(i+1)}_{\rho, j+1})-c'(\alg^{(i+1)}_{\rho, j})}{c'(\sol)}\\
&=\frac{(4\Gamma - 1)(\sqrt{\Gamma}-1)(\sqrt{\Gamma}-1-\epsilon)}{16(1+\epsilon)\Gamma^2} c(\alg^{(i+1)}_{\rho, J} \setminus \sol)  \frac{c'(\alg^{(i+1)}_{\rho, J})-c'(\alg^{(i+1)}_{\rho, 0})}{c'(\sol)}\\
&\geq \frac{(4\Gamma - 1)(\sqrt{\Gamma}-1)(\sqrt{\Gamma}-1-\epsilon)\kappa}{16(1+\epsilon)\Gamma^2} c(\alg^{(i+1)}_{\rho, J} \setminus \sol).
\end{align*}
The third-to-last inequality is using Observation~\ref{obs:swap}, which implies $c(\alg^{(i+1)}_{\rho, j} \setminus \sol) - c(\sol \setminus \alg^{(i+1)}_{\rho, j})$ is decreasing with swaps. The second-to-last inequality is using the assumption $c(\alg^{(i+1)}_\rho \setminus \sol) > 4\Gamma \cdot c(\sol \setminus \alg^{(i+1)}_{\rho})$ in the statement the lemma. The last inequality uses the fact that the while loop on line~\ref{line:firstwhile} of \textsc{DoubleApprox} terminates because $c'(\alg^{(i+1)}_{\rho, J}) > (4\Gamma' + \kappa)\chi'$ (by the assumption that $c(\alg^{(i+1)}_{\rho, J}) > c(\alg^{(i)}) - \rho /2$), and lines~\ref{line:approx} and~\ref{line:secondwhile} of \textsc{DoubleApprox} give that  $c'(\alg^{(i+1)}_{\rho, 0}) \leq 4\Gamma' \chi'$.
\end{proof}

\begin{lemma}[Backward Phase Analysis]\label{lemma:backwardphase}
Fix any even $i+2$ in algorithm \textsc{DoubleApprox} and any value of $\rho$. Suppose all values of $\alg^{(i+2)}_{\rho}$ satisfy $c(\alg^{(i+2)}_{\rho} \setminus \sol) > 4\Gamma \cdot c(\sol \setminus \alg^{(i+2)}_{\rho})$. Let $T =$ $\frac{c'(\alg^{(i+1)}_{\rho}) - c'(\alg^{(i+2)}_{\rho})}{c'(\sol)}$. Then for 

$$\zeta' = \frac{4(1+\epsilon)\Gamma'}{(\sqrt{\Gamma'}-1)(\sqrt{\Gamma'}-1-\epsilon)(4\Gamma'-1)(4\Gamma-1)}, $$ 

the final values of $\alg^{(i+1)}_\rho, \alg^{(i+2)}_\rho$ satisfy

$$c(\alg^{(i+2)}_\rho) - c(\alg^{(i+1)}_\rho) \leq (e^{\zeta' T} - 1) \cdot c(\alg^{(i+1)}_\rho \setminus \sol).$$
\end{lemma}

\begin{proof}
Because $c'(\alg^{(i+2)}_\rho) > 4\Gamma'\chi'$ in a backwards phase and $\chi' \geq c'(\sol)$,  by Lemma~\ref{lemma:greedyswap} whenever \textsc{GreedySwap} is called on $\alg^{(i+2)}_\rho$ in line~\ref{line:secondswap} of \textsc{DoubleApprox}, at least one swap is possible.  Since all edge costs are multiples of $\frac{\epsilon}{n}\chi'$, and  the last argument to \textsc{GreedySwap} is $\frac{\epsilon}{n}\chi'$ (which lower bounds the decrease in $c'(\alg^{(i+2)}_\rho)$ due to any improving swap), \textsc{GreedySwap} always makes a swap.

Let $\alg^{(i+2)}_{\rho, j}$ denote the value of $\alg^{(i+2)}$ after $j$ calls to \textsc{GreedySwap} on $\alg^{(i+2)}$, and let $J$ be the total number of calls of \textsc{GreedySwap} on $\alg^{(i+2)}$. Then $\alg^{(i+2)}_{\rho, 0}$ is the final value of $\alg^{(i+1)}$ and the final value of $\alg^{(i+2)}$ is $\alg^{(i+2)}_{\rho, J}$.  We want to show that
$$c(\alg^{(i+2)}_{\rho, J}) - c(\alg^{(i+2)}_{\rho, 0}) = \sum_{j = 0}^{J-1}[c(\alg^{(i+2)}_{\rho, j+1}) - c(\alg^{(i+2)}_{\rho, j})] \leq (e^{\zeta' T} - 1)  c(\alg^{(i+2)}_{\rho, 0} \setminus \sol).$$

We bound each $c(\alg^{(i+2)}_{j+1}) - c(\alg^{(i+2)}_{j})$ term using Lemma~\ref{lemma:greedyswap} and Lemma~\ref{lemma:apxsp}.  Since  $c'(\alg^{(i+2)}_\rho) > 4\Gamma' c'(\sol)$ in a backwards phase, by Lemma~\ref{lemma:greedyswap} we know there exists a swap between $a \in \alg^{(i+2)}_{\rho, j}$ and $f \in \sol$ such that

$$\frac{(1+\epsilon)c(f) - c(a)}{c'(a) - (1+\epsilon)c'(f)} \leq \frac{4(1+\epsilon)\Gamma'}{(\sqrt{\Gamma'}-1)(\sqrt{\Gamma'}-1-\epsilon)} \cdot \frac{c(\sol \setminus \alg^{(i+2)}_{\rho, j})}{c'(\alg^{(i+2)}_{\rho, j} \setminus \sol)}.$$ 

By Lemma~\ref{lemma:apxsp}, we know that when $G'$ is set to the value in $[c(f), (1+\epsilon) \cdot c(f)]$ in line 2 of in \textsc{GreedySwap}, the algorithm finds a path $f'$ between the endpoints of $f$ such that $c'(f') \leq (1+\epsilon) c'(f)$ and $c(f') \leq (1+\epsilon) c(f)$. Thus $(a, f') \in swaps$ and we get that the swap $(a^*, f^*)$ chosen by the $(j+1)$th call to \textsc{GreedySwap} satisfies:

\begin{align*}
\frac{c(f^*) - c(a^*)}{c'(a^*) - c'(f^*)}
&\leq \frac{c(f') - c(a)}{c'(a) - c'(f)} 
\leq \frac{(1+\epsilon)c(f) - c(a)}{c'(a) - (1+\epsilon)c'(f)} \\
&\leq
\frac{4(1+\epsilon)\Gamma'}{(\sqrt{\Gamma'}-1)(\sqrt{\Gamma'}-1-\epsilon)} \cdot \frac{c(\sol \setminus \alg^{(i+2)}_{\rho, j})}{c'(\alg^{(i+2)}_{\rho, j} \setminus \sol)}\\
&\leq\frac{4(1+\epsilon)\Gamma'}{(\sqrt{\Gamma'}-1)(\sqrt{\Gamma'}-1-\epsilon)(4\Gamma'-1)(4\Gamma)} \cdot \frac{c(\alg^{(i+2)}_{\rho, j} \setminus \sol)}{c'(\sol)}.\\
\end{align*}

The last inequality is derived using the assumption  $c(\alg^{(i+2)}_{\rho} \setminus \sol) > 4\Gamma \cdot c(\sol \setminus \alg^{(i+2)}_{\rho})$ in the statement of the lemma, as well as the fact that for all $j < J$,  $c'(\alg^{(i+2)}_{\rho, j}) \geq 4\Gamma c'(\sol) \implies c'(\alg^{(i+2)}_{\rho, j} \setminus \sol) \geq c'(\alg^{(i+2)}_{\rho, j}) - c'(\sol) \geq (4\Gamma'-1) c'(\sol)$.
This in turn gives:
\begin{align}
c(\alg^{(i+2)}_{\rho, J}) - c(\alg^{(i+2)}_{\rho, 0}) &= \sum_{j = 0}^{J-1}[c(\alg^{(i+2)}_{\rho, j+1}) - c(\alg^{(i+2)}_{\rho, j})] \nonumber \\
&= \sum_{j = 0}^{J-1}\frac{c(\alg^{(i+2)}_{\rho, j+1}) - c(\alg^{(i+2)}_{\rho, j})}{c'(\alg^{(i+2)}_{\rho, j}) - c'(\alg^{(i+2)}_{\rho, j+1})} \cdot [c'(\alg^{(i+2)}_{\rho, j}) - c'(\alg^{(i+2)}_{\rho, j+1})] \nonumber \\
&\leq  \frac{4(1+\epsilon)\Gamma'}{(\sqrt{\Gamma'}-1)(\sqrt{\Gamma'}-1-\epsilon)(4\Gamma'-1)(4\Gamma)} \sum_{j = 0}^{J-1} [c(\alg^{(i+2)}_{\rho, j} \setminus \sol)] \nonumber\\
&\qquad\qquad\cdot \frac{c'(\alg^{(i+2)}_{\rho, j}) - c'(\alg^{(i+2)}_{\rho, j+1})}{c'(\sol)} \nonumber \\
&\leq \frac{4(1+\epsilon)\Gamma'}{(\sqrt{\Gamma'}-1)(\sqrt{\Gamma'}-1-\epsilon)(4\Gamma'-1)(4\Gamma-1)} \nonumber\\
&\qquad\qquad \cdot \sum_{j = 0}^{J-1}  [c(\alg^{(i+2)}_{\rho, j} \setminus \sol) - c(\sol \setminus \alg^{(i+2)}_{\rho, j})] \nonumber\\
&\qquad\qquad  \cdot\frac{c'(\alg^{(i+2)}_{\rho, j}) - c'(\alg^{(i+2)}_{\rho, j+1})}{c'(\sol)} \nonumber \\
&= \zeta' \sum_{j = 0}^{J-1} [c(\alg^{(i+2)}_{\rho, j} \setminus \sol) - c(\sol \setminus \alg^{(i+2)}_{\rho, j})] \cdot \frac{c'(\alg^{(i+2)}_{\rho, j}) - c'(\alg^{(i+2)}_{\rho, j+1})}{c'(\sol)}.\label{eq:swapbound2}
\end{align}
The last inequality is proved using the assumption $c(\alg^{(i+2)}_{\rho} \setminus \sol) > 4\Gamma \cdot c(\sol \setminus \alg^{(i+2)}_{\rho})$ in the statement of the lemma, which implies
\begin{align*}
c(\alg^{(i+2)}_{\rho, j} \setminus \sol) &= \frac{4\Gamma}{4\Gamma-1}c(\alg^{(i+2)}_{\rho, j} \setminus \sol) - \frac{1}{4\Gamma-1}c(\alg^{(i+2)}_{\rho, j} \setminus \sol)\\
&< \frac{4\Gamma}{4\Gamma-1}c(\alg^{(i+2)}_{\rho, j} \setminus \sol) - \frac{4\Gamma }{4\Gamma-1}c(\sol \setminus \alg^{(i+2)}_{\rho, j}).
\end{align*}
It now suffices to show 

$$\sum_{j = 0}^{J-1} [c(\alg^{(i+2)}_{\rho, j} \setminus \sol) - c(\sol \setminus \alg^{(i+2)}_{\rho, j})] \cdot \frac{c'(\alg^{(i+2)}_{\rho, j}) - c'(\alg^{(i+2)}_{\rho, j+1})}{c'(\sol)} \leq $$
$$\frac{e^{\zeta' T} - 1}{\zeta'} c(\alg^{(i+2)}_{\rho, 0} \setminus \sol).$$ 

To do so, we view the series of swaps as occurring over a continuous timeline, where for $j = 0, 1, \ldots J-1$ the $(j+1)$th swap takes time $\tau(j) = \frac{c'(\alg^{(i+2)}_{\rho, j}) - c'(\alg^{(i+2)}_{\rho, j+1})}{c'(\sol)}$, i.e., occurs from time $\sum_{j' < j} \tau(j')$ to time $\sum_{j' \leq j} \tau(j')$. The total time taken to perform all swaps in the sum is the total decrease in $c'$ across all swaps, divided by $c'(\sol)$, i.e., exactly $T$. Using this definition of time, let $\Phi(t)$ denote $c(\alg^{(i+2)}_{\rho, j} \setminus \sol) - c(\sol \setminus \alg^{(i+2)}_{\rho, j})$ for the value of $j$ satisfying $\Phi(t) \in [\sum_{j' < j} \tau(j'), \sum_{j' \leq j} \tau(j'))$. Using this definition, we get: 

$$ \sum_{j = 0}^{J-1} [c(\alg^{(i+2)}_{\rho, j} \setminus \sol) - c(\sol \setminus \alg^{(i+2)}_{\rho, j})] \cdot \frac{c'(\alg^{(i+2)}_{\rho, j}) - c'(\alg^{(i+2)}_{\rho, j+1})}{c'(\sol)} = \int_0^{\rightarrow T} \Phi(t)\ dt.$$

We conclude by claiming $\Phi(t) \leq e^{\zeta' t} c(\alg^{(i+2)}_{\rho, 0} \setminus \sol)$. Given this claim, we get:

$$\int_0^{\rightarrow T} \Phi(t)\ dt \leq c(\alg^{(i+2)}_{\rho, 0} \setminus \sol) \int_0^{\rightarrow T} e^{\zeta' t}\ dt = \frac{e^{\zeta' T} - 1}{\zeta'}c(\alg^{(i+2)}_{\rho, 0} \setminus \sol).$$

Which completes the proof of the Lemma. We now focus on proving the claim. Since $\Phi(t)$ is fixed in the interval $[\sum_{j' < j} \tau(j'), \sum_{j' \leq j} \tau(j'))$, it suffices to prove the claim only for $t$ which are equal to $\sum_{j' < j} \tau(j')$ for some $j$, so we proceed by induction on $j$. The claim clearly holds for $j = 0$ since $\sum_{j' < 0} \tau(j') = 0$ and $\Phi(0) = c(\alg^{(i+2)}_{\rho, 0} \setminus \sol) - c(\sol \setminus \alg^{(i+2)}_{\rho, 0}) \leq c(\alg^{(i+2)}_{\rho, 0} \setminus \sol)$. 

Assume that for $t' = \sum_{j' < j} \tau(j')$, we have $\Phi(t') \leq e^{\zeta' t'} c(\alg^{(i+2)}_{\rho, 0} \setminus \sol)$. For $t'' = t' + \tau(j)$, by induction we can prove the claim by showing  $\Phi(t'') \leq e^{\zeta' \tau(j)} \Phi(t')$. 

To show this, we consider the quantity 

\begin{align*}
\Phi(t'') - \Phi(t') &= 
[c(\alg^{(i+2)}_{\rho, j+1} \setminus \sol) - c(\sol \setminus \alg^{(i+2)}_{\rho, j+1})] - [c(\alg^{(i+2)}_{\rho, j} \setminus \sol) - c(\sol \setminus \alg^{(i+2)}_{\rho, j})] \\
&= [c(\alg^{(i+2)}_{\rho, j+1} \setminus \sol) - c(\alg^{(i+2)}_{\rho, j} \setminus \sol)] + [c(\sol \setminus \alg^{(i+2)}_{\rho, j}) - c(\sol \setminus \alg^{(i+2)}_{\rho, j+1})].
\end{align*}

By Observation~\ref{obs:swap} and reusing the bound in \eqref{eq:swapbound2}, we have:

\begin{align*}
\Phi(t'') - \Phi(t') &= c(\alg^{(i+2)}_{\rho,j+1}) - c(\alg_{\rho,j}^{(i+2)}) \\
&\leq  \zeta' \frac{[c(\alg^{(i+2)}_{\rho, j} \setminus \sol) - c(\sol \setminus \alg^{(i+2)}_{\rho, j})]}{c'(\sol)} [c'(\alg^{(i+2)}_{\rho, j}) - c'(\alg^{(i+2)}_{\rho, j+1})]\\
&= \zeta' \cdot [c(\alg^{(i+2)}_{\rho, j} \setminus \sol) - c(\sol \setminus \alg^{(i+2)}_{\rho, j})] \cdot \tau(j) = \zeta' \cdot \Phi(t') \cdot \tau(j).
\end{align*}

Rearranging terms we have:
$$\Phi(t'') \leq \left(1 + \zeta' \cdot \tau(j)\right)\Phi(t')\leq e^{\zeta' \tau(j)} \Phi(t'),$$
where we use the inequality $1+x \leq e^x$. This completes the proof of the claim.
\end{proof}

\begin{corollary}\label{cor:swapgain}
Fix any positive even value of $i+2$ in algorithm \textsc{DoubleApprox}, and let $\rho$ be the power of $(1+\epsilon)$ times $\min_e c_e$ such that $\rho \in [c(\alg^{(i)} \setminus \sol), (1+\epsilon)c(\alg^{(i)} \setminus \sol)]$.  Suppose all values of $\alg^{(i+1)}_\rho$ and the final value of $\alg^{(i)}$ in \textsc{DoubleApprox} satisfy $c(\alg^{(i+1)}_\rho \setminus \sol) > 4\Gamma \cdot c(\sol \setminus \alg^{(i+1)}_{\rho})$ and $c(\alg^{(i)} \setminus \sol) > 4\Gamma \cdot c(\sol \setminus \alg^{(i)})$. Then for $0 < \epsilon < 2/3 -5/12\Gamma$ and $\zeta'$ as defined in Lemma~\ref{lemma:backwardphase}, the final values of $\alg^{(i+2)}, \alg^{(i)}$ satisfy

$$c(\alg^{(i)}) - c(\alg^{(i+2)}) \geq$$
$$\left[\min\left\{\frac{4\Gamma - 1}{8\Gamma} ,\frac{(4\Gamma - 1)(\sqrt{\Gamma}-1)(\sqrt{\Gamma}-1-\epsilon)\kappa}{16(1+\epsilon)\Gamma^2}\right\} - (e^{\zeta' (4\Gamma'+\kappa+1+\epsilon)} - 1)\right] \cdot c(\alg^{(i+1)}_\rho \setminus \sol).$$

\end{corollary}

\begin{proof}
It suffices to lower bound $c(\alg^{(i)}) - c(\alg^{(i+2)}_\rho)$ for this value of $\rho$, since $c(\alg^{(i)}) - c(\alg^{(i+2)})$ must be at least this large. After rescaling $\epsilon$ appropriately, we have
$$c'(\alg^{(i+1)}_\rho) - c'(\alg^{(i+2)}_\rho) \leq c'(\alg^{(i+1)}_\rho) \leq (4\Gamma' + \kappa + 1 + \epsilon)c'(\sol),$$
because the algorithm can increase its cost with respect to $c'$ by at most $(1+\epsilon)c'(\sol)$ in any swap in the forward phase (by line~\ref{line:increasebound} of \textsc{GreedySwap}, which bounds the increase $w'(\hat{f}) \leq W \leq (1+\epsilon)w'(\sol)$), so it exceeds the threshold $(4\Gamma'+\kappa)\chi' \leq (4\Gamma'+\kappa)(1+\epsilon)c'(\sol)$ on line~\ref{line:secondwhile} of \textsc{DoubleApprox} by at most this much. Then applying Lemma~\ref{lemma:forwardphase} to $c(\alg^{(i)}) - c(\alg^{(i+1)}_\rho)$ and Lemma~\ref{lemma:backwardphase} to $c(\alg^{(i+1)}_\rho) - c(\alg^{(i+2)}_\rho)$ (using $T \leq 4\Gamma' + \kappa + 1 +\epsilon$) gives:

$$c(\alg^{(i)}) - c(\alg^{(i+2)}_\rho) = [c(\alg^{(i)}) - c(\alg^{(i+1)}_\rho)] + [c(\alg^{(i+1)}_\rho) - c(\alg^{(i+2)}_\rho)]$$
$$\geq \left[\min\left\{\frac{4\Gamma - 1}{8\Gamma} ,\frac{(4\Gamma - 1)(\sqrt{\Gamma}-1)(\sqrt{\Gamma}-1-\epsilon)\kappa}{16(1+\epsilon)\Gamma^2}\right\} - (e^{\zeta' (4\Gamma'+\kappa+1+\epsilon)} - 1)\right] \cdot c(\alg^{(i+1)}_\rho \setminus \sol).$$
\end{proof}

\begin{lemma}\label{lemma:algconverges}
Suppose $\Gamma, \Gamma', \kappa$, and $\epsilon$ are chosen such that for $\zeta'$ as defined in Lemma~\ref{lemma:backwardphase}, 
$$\min\left\{\frac{4\Gamma - 1}{8\Gamma} ,\frac{(4\Gamma - 1)(\sqrt{\Gamma}-1)(\sqrt{\Gamma}-1-\epsilon)\kappa}{16(1+\epsilon)\Gamma^2}\right\} - (e^{\zeta' (4\Gamma'+\kappa+1+\epsilon)} - 1) > 0,$$
and $0 < \epsilon < 2/3 -5/12\Gamma$. Let $\eta$ equal
$$\frac{\min\left\{\frac{4\Gamma - 1}{8\Gamma} ,\frac{(4\Gamma - 1)(\sqrt{\Gamma}-1)(\sqrt{\Gamma}-1-\epsilon)\kappa}{16(1+\epsilon)\Gamma^2}\right\} - (e^{\zeta' (4\Gamma'+\kappa+1+\epsilon)} - 1)}{1 + \frac{4\Gamma - 1}{4\Gamma} (e^{\zeta' (4\Gamma'+\kappa+1+\epsilon)} - 1)}.$$
Assume $\eta > 0$ and let $I = 2(\lceil  \log \frac{n \max_e c_e}{\min_e c_e} / \log (1+\eta) \rceil+1)$. Then there exists some intermediate value $\alg^*$ assigned to $\alg^{(i)}_\rho$ by the algorithm for some $i \leq I$ and $\rho$ such that $c(\alg^* \setminus \sol) \leq 4 \Gamma c(\sol \setminus \alg^*)$ and $c'(\alg^*) \leq (4\Gamma' +\kappa+ 1 + \epsilon) c'(\sol)$.
\end{lemma}
\begin{proof}
Let $\Phi(i) := c(\alg^{(i)} \setminus \sol) - c(\sol \setminus \alg^{(i)})$ for even $i$. Assume that the lemma is false.   Since algorithm \textsc{DoubleApprox} guarantees that $c'(\alg^{(i)}_{\rho}) \leq (4\Gamma' +\kappa+ 1 + \epsilon) c'(\sol)$, if the lemma is false it must be that for all $i$ and $\rho$, $c(\alg^{(i)}_{\rho} \setminus \sol) > 4 \Gamma c(\sol \setminus \alg^{(i)}_{\rho})$. By Corollary~\ref{cor:swapgain}, and the assumption $\min\left\{\frac{4\Gamma - 1}{8\Gamma} ,\frac{(4\Gamma - 1)(\sqrt{\Gamma}-1)(\sqrt{\Gamma}-1-\epsilon)\kappa}{16(1+\epsilon)\Gamma^2}\right\} - (e^{\zeta' (4\Gamma'+\kappa+1+\epsilon)} - 1) > 0$ in the statement of this lemma, for all $i$  $c(\alg^{(i)}) < c(\alg^{(i-2)})$, so the while loop on Line~\ref{line:mainwhile} of \textsc{DoubleApprox} never breaks. This means that for all even $i \leq I$, $\alg^{(i)}$ is assigned a value in \textsc{DoubleApprox}. We will show that this implies that for the final value of $\alg^{(I)}$, $\Phi(I) = c(\alg^{(I)} \setminus \sol) - c(\sol \setminus \alg^{(I)}) < \frac{4\Gamma - 1}{4\Gamma} \min_e c_e$.  In the proof of Lemma~\ref{lemma:backwardphase}, we showed how the inequality $c(\alg^{(I)} \setminus \sol) > 4 \Gamma c(\sol \setminus \alg^{(I)})$ could be used to prove
$c(\alg^{(I)} \setminus \sol) - c(\sol \setminus \alg^{(I)}) > \frac{4\Gamma - 1}{4\Gamma} c(\alg^{(I)} \setminus \sol).$
The value of $c(\alg^{(I)} \setminus \sol)$ must be positive (otherwise $c(\alg^{(I)} \setminus \sol) \leq 4\Gamma c(\sol \setminus \alg^{(I)})$ trivially), and hence it must be at least $\min_e c_e$.  These two inequalities conflict, which implies a contradiction.  Hence the lemma must be true.

We now analyze how the quantity $\Phi(i)$ changes under the assumption that the lemma is false. Of course $\Phi(0) \leq n \max_e c_e$. Observation~\ref{obs:swap} gives that $\Phi(i) - \Phi(i+2)$ is exactly equal to $c(\alg^{(i)}) - c(\alg^{(i+2)})$. For the value of $\rho$ such that $\rho \in [c(\alg^{(i)} \setminus \sol), (1+\epsilon) c(\alg^{(i)} \setminus \sol)]$, by Corollary~\ref{cor:swapgain} and the assumption that the lemma is false, for even $i$ we have 

$$\Phi(i) - \Phi(i+2) $$
$$\geq \left[\min\left\{\frac{4\Gamma - 1}{8\Gamma} ,\frac{(4\Gamma - 1)(\sqrt{\Gamma}-1)(\sqrt{\Gamma}-1-\epsilon)\kappa}{16(1+\epsilon)\Gamma^2}\right\} - (e^{\zeta' (4\Gamma'+\kappa+1+\epsilon)} - 1)\right] \cdot c(\alg^{(i+1)}_\rho \setminus \sol) $$
$$\geq \left[\min\left\{\frac{4\Gamma - 1}{8\Gamma} ,\frac{(4\Gamma - 1)(\sqrt{\Gamma}-1)(\sqrt{\Gamma}-1-\epsilon)\kappa}{16(1+\epsilon)\Gamma^2}\right\} - (e^{\zeta' (4\Gamma'+\kappa+1+\epsilon)} - 1)\right] $$
\begin{equation}\label{eq:algconverges}\cdot [c(\alg^{(i+1)}_\rho \setminus \sol) - c(\sol \setminus \alg^{(i+1)}_\rho)].
\end{equation}
Lemma~\ref{lemma:backwardphase} (using the proof from Corollary~\ref{cor:swapgain} that $T \leq 4\Gamma' + \kappa + 1 +\epsilon$), Observation~\ref{obs:swap}, and the inequality $c(\alg^{(i+1)}_{\rho} \setminus \sol) - c(\sol \setminus \alg^{(i+1)}_{\rho}) > \frac{4\Gamma - 1}{4\Gamma} c(\alg^{(i+1)}_{\rho} \setminus \sol)$ give:

\begin{align*}
&\Phi(i+2) - [c(\alg^{(i+1)}_\rho \setminus \sol) - c(\sol \setminus \alg^{(i+1)}_\rho)] \\
\leq &(e^{\zeta' (4\Gamma'+\kappa+1+\epsilon)} - 1)]) c(\alg^{(i+1)}_\rho \setminus \sol) \\
<& \frac{4\Gamma - 1}{4\Gamma} (e^{\zeta' (4\Gamma'+\kappa+1+\epsilon)} - 1)]) [c(\alg^{(i+1)}_\rho \setminus \sol) - c(\sol \setminus \alg^{(i+1)})]\\
\implies &[c(\alg^{(i+1)}_\rho \setminus \sol) - c(\sol \setminus \alg^{(i+1)}_\rho)] > \frac{1}{1 + \frac{4\Gamma - 1}{4\Gamma} (e^{\zeta' (4\Gamma'+\kappa+1+\epsilon)} - 1))}\Phi(i+2).
\end{align*}

Plugging this into \eqref{eq:algconverges} gives:

$$\Phi(i+2) < \left(1 +  \frac{\min\{\frac{4\Gamma - 1}{8\Gamma} ,\frac{(4\Gamma - 1)(\sqrt{\Gamma}-1)(\sqrt{\Gamma}-1-\epsilon)\kappa}{16(1+\epsilon)\Gamma^2}\} - (e^{\zeta' (4\Gamma'+\kappa+1+\epsilon)} - 1)}{1 + \frac{4\Gamma - 1}{4\Gamma} (e^{\zeta' (4\Gamma+\kappa+1+\epsilon)} - 1)}\right)^{-1} \Phi(i) $$
$$= (1+\eta)^{-1}\Phi(i).$$

Applying this inductively gives:

$$\Phi(i) \leq (1+\eta)^{-i/2} \Phi(0) \leq (1+\eta)^{-i/2} n \max_e c_e.$$

Plugging in $i = I =  2(\lceil  \log \frac{n \max_e c_e}{\min_e c_e} / \log (1+\eta) \rceil+1)$ gives $\Phi(I) \leq (1 + \eta)^{-1} \min_e c_e < \min_e c_e$ as desired.
\end{proof}

\begin{proof}[Proof of Lemma~\ref{lemma:differenceapprox2}]
If we have  $\chi \in [c(\sol), (1+\epsilon)\cdot c(\sol)]$ and $\chi' \in [c'(\sol), (1+\epsilon)\cdot c'(\sol)]$,
and the $c_e'$ values are multiples of $\frac{\epsilon}{n}\chi'$, then the conditions of \textsc{DoubleApprox} are met. As long as \eqref{eq:maincond} holds, that is:

\begin{equation*}
\min\left\{\frac{4\Gamma - 1}{8\Gamma} ,\frac{(4\Gamma - 1)(\sqrt{\Gamma}-1)(\sqrt{\Gamma}-1-\epsilon)\kappa}{16(1+\epsilon)\Gamma^2}\right\} - (e^{\zeta' (4\Gamma'+\kappa+1+\epsilon)} - 1) > 0,\tag{\ref{eq:maincond}}
\end{equation*}
then we have $\eta > 0$ in Lemma~\ref{lemma:algconverges}, thus giving the approximation guarantee in Lemma~\ref{lemma:differenceapprox2}. For any positive $\epsilon, \kappa, \Gamma'$, there exists a sufficiently large value of $\Gamma$ for \eqref{eq:maincond} to hold, since as $\Gamma \rightarrow \infty$, we have $\zeta' \rightarrow 0$, $(e^{\zeta' (4\Gamma'+\kappa+1+\epsilon)} - 1) \rightarrow 0$, and $\min\left\{\frac{4\Gamma - 1}{8\Gamma} ,\frac{(4\Gamma - 1)(\sqrt{\Gamma}-1)(\sqrt{\Gamma}-1-\epsilon)\kappa}{16(1+\epsilon)\Gamma^2}\right\} \rightarrow \min\{1/2, \kappa/(4+4\epsilon)\}$, so for any fixed choice of $\epsilon, \kappa, \Gamma'$, a sufficiently large value of $\Gamma$ causes $\eta > 0$ to hold as desired.

Some value in $\{\min_e c_e, (1+\epsilon)\min_e c_e, \ldots (1+\epsilon)^{\lceil \log_{1+\epsilon}\frac{n \max_e c_e}{\min_e c_e} \rceil} \min_e c_e\}$ satisfies the conditions for $\chi$, and there are polynomially many values in this set. The same holds for $\chi'$ in $\{\min_e c'_e, (1+\epsilon)\min_e c'_e, \ldots (1+\epsilon)^{\lceil \log_{1+\epsilon}\frac{n \max_e c'_e}{\min_e c'_e} \rceil} \min_e c'_e\}$. So we can run $\textsc{DoubleApprox}$ for all pairs of $\chi, \chi'$ (paying a polynomial increase in runtime), and output the union of all outputs, giving the guarantee of Lemma~\ref{lemma:differenceapprox2} by Lemma~\ref{lemma:algconverges}. For each $\chi'$ we choose, we can round the edge costs to the nearest multiple of $\frac{\epsilon}{n}\chi'$ before running \textsc{DoubleApprox}, and by Observation~\ref{obs:rounding} we only pay an additive $O(\epsilon)$ in the approximation factor with respect to $c'$. Finally, we note that by setting $\epsilon$ appropriately in the statement of Lemma~\ref{lemma:algconverges}, we can achieve the approximation guarantee stated in Lemma~\ref{lemma:differenceapprox2} for a different value of $\epsilon$. 

Then, we just need to show \textsc{DoubleApprox} runs in polynomial time. Lemma~\ref{lemma:algconverges} shows that the while loop of Line~\ref{line:mainwhile} only needs to be run a polynomial number ($I$) of times. The while loop for the forward phase runs at most $O(n^2)$ times since each call to \textsc{GreedySwap} decreases the cost with respect to $c$ by at least $\frac{1}{10n^2} \rho$, and once the total decrease exceeds $\rho/2$ the while loop breaks. The while loop for the backward phase runs at most $(\kappa+1+\epsilon)\frac{n}{\epsilon}$ times, since the initial cost with respect to $c'$ is at most $(4\Gamma + \kappa + 1 + \epsilon)\chi'$, the while loop breaks when it is less than $4\Gamma'\chi'$, and each call to \textsc{GreedySwap} improves the cost by at least $\frac{\epsilon}{n} \chi'$. Lastly, \textsc{GreedySwap} can be run in polynomial time as the maximal $a$ which need to be enumerated can be computed in polynomial time as described in Section~\ref{section:localsearch}.
\end{proof}
\section{Hardness Results for Robust Problems}\label{sec:hardness}

We give the following general hardness result for a family of problems that includes many graph optimization problems:

\begin{theorem}\label{thm:hardness}
Let $\mathcal{P}$ be any robust covering problem whose input includes a weighted graph $G$ where the lengths $d_e$ of the edges are given as ranges $[\ell_e, u_e]$ and for which the non-robust version of the problem, $\mathcal{P}'$, has the following properties:

\begin{itemize}
\item A solution to an instance of $\mathcal{P}'$ can be written as a (multi-)set $S$ of edges in $G$, and has cost $\sum_{e \in S} d_e$.
\item Given an input including $G$ to $\mathcal{P}'$, there is a polynomial-time approximation-preserving reduction from solving $\mathcal{P}'$ on this input to solving $\mathcal{P}'$ on some input including $G'$, where $G'$ is the graph formed by taking $G$, adding a new vertex $v^*$, and adding a single edge from $v^*$ to some $v \in V$ of weight 0.
\item For any input including $G$ to $\mathcal{P}'$, given any spanning tree $T$ of $G$, there exists a feasible solution only including edges from $T$.
\end{itemize}

Then, if there exists a polynomial time $(\alpha, \beta)$-robust algorithm for $\mathcal{P}$, there exists a polynomial-time $\beta$-approximation algorithm for $\mathcal{P}$.
\end{theorem}

Before proving Theorem~\ref{thm:hardness}, we note that robust traveling salesman and robust Steiner tree are examples of problems that Theorem~\ref{thm:hardness} implicitly gives lower bounds for. For both problems, the first property clearly holds.

For traveling salesman, given any input $G$, any solution to the problem on input $G'$ as described in Theorem~\ref{thm:hardness} can be turned into a solution of the same cost on input $G$ by removing the new vertex $v^*$ (since $v^*$ was distance 0 from $v$, removing $v^*$ does not affect the length of any tour), giving the second property. For any spanning tree of $G$, a walk on the spanning tree gives a valid TSP tour, giving the third property.

For Steiner tree, for the input with graph $G'$ and the same terminal set, for any solution containing the edge $(v, v^*)$ we can remove this edge and get a solution for the input with graph $G$ that is feasible and of the same cost. Otherwise, the solution is already a solution for the input with graph $G$ that is feasible and of the same cost, so the second property holds. Any spanning tree is a feasible Steiner tree, giving the third property.

We now give the proof of Theorem~\ref{thm:hardness}.

\begin{proof}[Proof of Theorem~\ref{thm:hardness}]
Suppose there exists a polynomial time $(\alpha, \beta)$-robust algorithm $A$ for $\mathcal{P}$. The $\beta$-approximation algorithm for $\mathcal{P}'$ is as follows:

\begin{enumerate}
\item From the input instance $\mathcal{I}$ of $\mathcal{P}$ where the graph is $G$, use the approximation-preserving reduction (that must exist by the second property of the theorem) to construct instance $\mathcal{I}'$ of $\mathcal{P'}$ where the graph is $G'$.
\item Construct an instance $\mathcal{I}''$ of $\mathcal{P}$ from $\mathcal{I}'$ as follows: For all edges in $G'$, their length is fixed to their length in $\mathcal{I'}$. In addition, we add a ``special'' edge from $v^*$ to all vertices besides $v$ with length range $[0, \infty]$\footnote{ $\infty$ is used to simplify the proof, but can be replaced with a sufficiently large finite number. For example, the total weight of all edges in $G$ suffices and has small bit complexity.}.
\item Run $A$ on $\mathcal{I}''$ to get a solution $\sol$. Treat this solution as a solution to $\mathcal{I}'$ (we will show it only uses edges that appear in $\mathcal{I}$). Use the approximation-preserving reduction to convert $\sol$ into a solution for $\mathcal{I}$ and output this solution.
\end{enumerate}

Let $O$ denote the cost of the optimal solution to $\mathcal{I}'$. Then, $\mr \leq O$. To see why, note that the optimal solution to $\mathcal{I}'$ has cost $O$ in all realizations of demands since it only uses edges of fixed cost, and thus its regret is at most $O$. This also implies that for all $\bd$, $\opt(\bd)$ is finite. Then for all $\bd$, $\sol(\bd) \leq \alpha \cdot \opt(\bd) + \beta \cdot \mr$, i.e. $\sol(\bd)$ is finite in all realizations of demands, so $\sol$ does not include any special edges, as any solution with a special edge has infinite cost in some realization of demands.

Now consider the realization of demands $\bd$ where all special edges have length 0. The special edges and the edge $(v, v^*)$ span $G'$, so by the third property of $\mathcal{P}'$ in the theorem statement there is a solution using only cost $0$ edges in this realization, i.e. $\opt(\bd)$ = 0. Then in this realization, $\sol(\bd) \leq \alpha \cdot \opt(\bd) + \beta \cdot \mr \leq \beta \cdot O$. But since $\sol$ does not include any special edges, and all edges besides special edges have fixed cost and their cost is the same in $\mathcal{I}''$ as in $\mathcal{I}'$, $\sol(\bd)$ also is the cost of $\sol$ in instance $\mathcal{I}'$, i.e. $\sol(\bd)$ is a $\beta$-approximation for $\mathcal{I}'$. Since the reduction from $\mathcal{I}$ to $\mathcal{I}'$ is approximation-preserving, we also get a $\beta$-approximation for $\mathcal{I}$.

\end{proof}

From \cite{MiroslavJMM02, KarpinskiLS13} we then get the following hardness results:

\begin{corollary}
Finding an $(\alpha, \beta)$-robust solution for Steiner tree where $\beta < 96/95$ is NP-hard.
\end{corollary}

\begin{corollary}
Finding an $(\alpha, \beta)$-robust solution for TSP where $\beta < 121/120$ is NP-hard.
\end{corollary}
\section{Conclusion}
\label{sec:conclusion}

In this paper, we designed constant approximation algorithms for the robust Steiner tree (\stt) and traveling salesman problems (\tsp). More precisely, our algorithms take as input a range of possible edge lengths in a graph and obtain a single solution for the problem at hand that can be compared to the optimal solution for any realization of edge lengths in the given ranges. While our approximation bounds for \tsp are small constants, that for \stt are very large constants. A natural question is whether these constants can be made smaller, e.g., of the same scale as classic approximation bounds for \stt. While we did not seek to optimize our constants, obtaining truly small constants for \stt appears to be beyond our techniques, and is an interesting open question. 

More generally, robust algorithms are a key component in the area of optimization under uncertainty that is of much practical and theoretical significance. Indeed, as mentioned in our survey of related work, several different models of robust algorithms have been considered in the literature. 
Optimizing over input ranges is one of the most natural models in robust optimization, but has been restricted in the past to polynomial-time solvable problems because of definitional limitations. We circumvent this by setting regret minimization as our goal, and creating the $(\alpha, \beta)$-approximation framework, which then allows us to consider a large variety of interesting combinatorial optimization problems in this setting.
We hope that our work will lead to more research in robust algorithms for other fundamental problems in combinatorial optimization, particularly in algorithmic graph theory.


\bibliographystyle{abbrv}
\bibliography{ref}

\begin{thebibliography}{10}

\bibitem{AissiBV08}
H.~Aissi, C.~Bazgan, and D.~Vanderpooten.
\newblock Complexity of the min–max (regret) versions of min cut problems.
\newblock {\em Discrete Optimization}, 5(1):66 -- 73, 2008.

\bibitem{AISSI2009427}
H.~Aissi, C.~Bazgan, and D.~Vanderpooten.
\newblock Min–max and min–max regret versions of combinatorial optimization
  problems: A survey.
\newblock {\em European Journal of Operational Research}, 197(2):427 -- 438,
  2009.

\bibitem{Averbakh01}
I.~Averbakh.
\newblock On the complexity of a class of combinatorial optimization problems
  with uncertainty.
\newblock {\em Mathematical Programming}, 90(2):263--272, Apr 2001.

\bibitem{Averbakh05}
I.~Averbakh.
\newblock The minmax relative regret median problem on networks.
\newblock {\em {INFORMS} Journal on Computing}, 17(4):451--461, 2005.

\bibitem{AverbakhB97}
I.~Averbakh and O.~Berman.
\newblock Minimax regret p-center location on a network with demand
  uncertainty.
\newblock {\em Location Science}, 5(4):247 -- 254, 1997.

\bibitem{AverbakhB00}
I.~Averbakh and O.~Berman.
\newblock Minmax regret median location on a network under uncertainty.
\newblock {\em INFORMS Journal on Computing}, 12(2):104--110, 2000.

\bibitem{BertsimasS2003}
D.~Bertsimas and M.~Sim.
\newblock Robust discrete optimization and network flows.
\newblock {\em Mathematical Programming}, 98(1):49--71, Sep 2003.

\bibitem{ByrkaGRS10}
J.~Byrka, F.~Grandoni, T.~Rothvo{\ss}, and L.~Sanit{\`{a}}.
\newblock An improved {LP}-based approximation for {Steiner} tree.
\newblock In {\em Proceedings of the 42nd {ACM} Symposium on Theory of
  Computing, {STOC} 2010, Cambridge, Massachusetts, USA, 5-8 June 2010}, pages
  583--592, 2010.

\bibitem{CharikarCP05}
M.~Charikar, C.~Chekuri, and M.~P\'{a}l.
\newblock Sampling bounds for stochastic optimization.
\newblock In {\em Proceedings of the 8th International Workshop on
  Approximation, Randomization and Combinatorial Optimization Problems, and
  Proceedings of the 9th International Conference on Randamization and
  Computation: Algorithms and Techniques}, APPROX'05/RANDOM'05, pages 257--269,
  Berlin, Heidelberg, 2005. Springer-Verlag.

\bibitem{ChasseinG15}
A.~Chassein and M.~Goerigk.
\newblock On the recoverable robust traveling salesman problem.
\newblock {\em Optimization Letters}, 10, 09 2015.

\bibitem{Chekuri07}
C.~Chekuri.
\newblock Routing and network design with robustness to changing or uncertain
  traffic demands.
\newblock {\em {SIGACT} News}, 38(3):106--129, 2007.

\bibitem{MiroslavJMM02}
M.~Chleb{\'i}k and J.~Chleb{\'i}kov{\'a}.
\newblock Approximation hardness of the {Steiner} tree problem on graphs.
\newblock In M.~Penttonen and E.~M. Schmidt, editors, {\em Algorithm Theory ---
  SWAT 2002}, pages 170--179, Berlin, Heidelberg, 2002. Springer Berlin
  Heidelberg.

\bibitem{Conde12}
E.~Conde.
\newblock On a constant factor approximation for minmax regret problems using a
  symmetry point scenario.
\newblock {\em European Journal of Operational Research}, 219(2):452 -- 457,
  2012.

\bibitem{DhamdhereGRS05}
K.~Dhamdhere, V.~Goyal, R.~Ravi, and M.~Singh.
\newblock How to pay, come what may: Approximation algorithms for demand-robust
  covering problems.
\newblock In {\em 46th Annual {IEEE} Symposium on Foundations of Computer
  Science {(FOCS} 2005), 23-25 October 2005, Pittsburgh, PA, USA, Proceedings},
  pages 367--378, 2005.

\bibitem{FeigeJMM07}
U.~Feige, K.~Jain, M.~Mahdian, and V.~S. Mirrokni.
\newblock Robust combinatorial optimization with exponential scenarios.
\newblock In {\em Integer Programming and Combinatorial Optimization, 12th
  International {IPCO} Conference, Ithaca, NY, USA, June 25-27, 2007,
  Proceedings}, pages 439--453, 2007.

\bibitem{GoyalOS13}
N.~Goyal, N.~Olver, and F.~B. Shepherd.
\newblock The {VPN} conjecture is true.
\newblock {\em J. {ACM}}, 60(3):17:1--17:17, 2013.

\bibitem{GGKMSSV17}
M.~Gro{\ss}, A.~Gupta, A.~Kumar, J.~Matuschke, D.~R. Schmidt, M.~Schmidt, and
  J.~Verschae.
\newblock A local-search algorithm for {Steiner} forest.
\newblock In {\em 9th Innovations in Theoretical Computer Science Conference,
  {ITCS} 2018, January 11-14, 2018, Cambridge, MA, {USA}}, pages 31:1--31:17,
  2018.

\bibitem{GuptaKKRY01}
A.~Gupta, J.~M. Kleinberg, A.~Kumar, R.~Rastogi, and B.~Yener.
\newblock Provisioning a virtual private network: a network design problem for
  multicommodity flow.
\newblock In {\em Proceedings on 33rd Annual {ACM} Symposium on Theory of
  Computing, July 6-8, 2001, Heraklion, Crete, Greece}, pages 389--398, 2001.

\bibitem{GuptaNR14}
A.~Gupta, V.~Nagarajan, and R.~Ravi.
\newblock Thresholded covering algorithms for robust and max-min optimization.
\newblock {\em Math. Program.}, 146(1-2):583--615, 2014.

\bibitem{GuptaNR16}
A.~Gupta, V.~Nagarajan, and R.~Ravi.
\newblock Robust and maxmin optimization under matroid and knapsack uncertainty
  sets.
\newblock {\em {ACM} Trans. Algorithms}, 12(1):10:1--10:21, 2016.

\bibitem{GuptaPRS04}
A.~Gupta, M.~P\'{a}l, R.~Ravi, and A.~Sinha.
\newblock Boosted sampling: Approximation algorithms for stochastic
  optimization.
\newblock In {\em Proceedings of the Thirty-sixth Annual ACM Symposium on
  Theory of Computing}, STOC '04, pages 417--426, New York, NY, USA, 2004. ACM.

\bibitem{InuiguchiM95}
M.~Inuiguchi and M.~Sakawa.
\newblock Minimax regret solution to linear programming problems with an
  interval objective function.
\newblock {\em European Journal of Operational Research}, 86(3):526 -- 536,
  1995.

\bibitem{KarpinskiLS13}
M.~Karpinski, M.~Lampis, and R.~Schmied.
\newblock New inapproximability bounds for {TSP}.
\newblock In L.~Cai, S.-W. Cheng, and T.-W. Lam, editors, {\em Algorithms and
  Computation}, pages 568--578, Berlin, Heidelberg, 2013. Springer Berlin
  Heidelberg.

\bibitem{KasperskiZ06}
A.~Kasperski and P.~Zieli\'{n}ski.
\newblock An approximation algorithm for interval data minmax regret
  combinatorial optimization problems.
\newblock {\em Inf. Process. Lett.}, 97(5):177--180, Mar. 2006.

\bibitem{KasperskiZ07}
A.~Kasperski and P.~Zieli\'nski.
\newblock On the existence of an {FPTAS} for minmax regret combinatorial
  optimization problems with interval data.
\newblock {\em Oper. Res. Lett.}, 35:525--532, 2007.

\bibitem{KhandekarKMS13}
R.~Khandekar, G.~Kortsarz, V.~S. Mirrokni, and M.~R. Salavatipour.
\newblock Two-stage robust network design with exponential scenarios.
\newblock {\em Algorithmica}, 65(2):391--408, 2013.

\bibitem{KouvelisY96}
P.~Kouvelis and G.~Yu.
\newblock {\em Robust Discrete Optimization and Its Applications}.
\newblock Springer US, 1996.

\bibitem{KouvelisY97}
P.~Kouvelis and G.~Yu.
\newblock {\em Robust 1-Median Location Problems: Dynamic Aspects and
  Uncertainty}, pages 193--240.
\newblock Springer US, Boston, MA, 1997.

\bibitem{MausserL98}
H.~E. Mausser and M.~Laguna.
\newblock A new mixed integer formulation for the maximum regret problem.
\newblock {\em International Transactions in Operational Research}, 5(5):389 --
  403, 1998.

\bibitem{ShmoysS06}
D.~B. Shmoys and C.~Swamy.
\newblock An approximation scheme for stochastic linear programming and its
  application to stochastic integer programs.
\newblock {\em J. ACM}, 53(6):978--1012, Nov. 2006.

\bibitem{SwamyS06}
C.~Swamy and D.~B. Shmoys.
\newblock Approximation algorithms for 2-stage stochastic optimization
  problems.
\newblock {\em {SIGACT} News}, 37(1):33--46, 2006.

\bibitem{SwamyS12}
C.~Swamy and D.~B. Shmoys.
\newblock Sampling-based approximation algorithms for multistage stochastic
  optimization.
\newblock {\em {SIAM} J. Comput.}, 41(4):975--1004, 2012.

\bibitem{Vazirani01}
V.~Vazirani.
\newblock {\em Approximation algorithms}.
\newblock Springer-Verlag, Berlin, 2001.

\bibitem{Vygen_newapproximation}
J.~Vygen.
\newblock New approximation algorithms for the tsp.

\bibitem{Wolsey80}
L.~A. Wolsey.
\newblock Heuristic analysis, linear programming and branch and bound.
\newblock In V.~J. Rayward-Smith, editor, {\em Combinatorial Optimization II},
  pages 121--134. Springer Berlin Heidelberg, Berlin, Heidelberg, 1980.

\bibitem{YamanKP01}
H.~Yaman, O.~E. Kara\c{s}an, and M.~\c{C}. Pinar.
\newblock The robust spanning tree problem with interval data.
\newblock {\em Operations Research Letters}, 29(1):31 -- 40, 2001.

\bibitem{Zielinski04}
P.~Zieli\'nski.
\newblock The computational complexity of the relative robust shortest path
  problem with interval data.
\newblock {\em European Journal of Operational Research}, 158(3):570 -- 576,
  2004.

\end{thebibliography}

\end{document}